\documentclass[runningheads, envcountsame, a4paper]{llncs}
\usepackage{amsfonts,amssymb,amsmath}
\usepackage[T2A]{fontenc}
\usepackage[utf8]{inputenc}        

\usepackage{microtype}

\usepackage[hyphens]{url}
\usepackage{verbatim}

\let\leq\leqslant
\let\geq\geqslant

\def\poly{\mathop{\mathrm{poly}}\nolimits}  

\def\leGen#1#2{\mathop{\leq^{\mathrm{#2}}_{\mathrm{#1}}}}

\def\lelog{\leGen{log}{}}

\def\lerat{\mathop{\leq_{\mathrm{rat}}}}

\def\nreg{\mathrm{NRR}}
\def\Per{\mathrm{Per}}
\def\Primes{\mathrm{PRIMES}}

\let\es\varnothing

\def\NN{\mathbb N}

\let\epsilon\varepsilon
\let\eps\varepsilon
\let\al\alpha
\def\bal{\boldsymbol{\alpha}}

\let\rendmarker\vartriangleleft

\def\R{{\cal R}}
\def\A{{\cal A}}
\def\B{{\cal B}}

\def\T{{\cal T}}

\def\0{\mathsf{false}}
\def\1{\mathsf{true}}
\def\start{\mathsf{start}}

\def\PP{{\mathbf{P}}}
\def\NP{{\mathbf{NP}}}
\def\REG{{\mathsf{REG}}}

\def\CFL{{\mathsf{CFL}}}

\def\DCFL{{\mathsf{DCFL}}}
\def\DSA{{\mathsf{DSA}}}
\def\SA{{\mathsf{SA}}}
\def\CVP{{\mathsf{CVP}}}
\def\SAT{{\mathsf{3\text{-}SAT}}}
\def\SASAT{{\mathsf{SA\text{-}SAT}}}
\def\SACVP{{\text{\sf{SA\text{-}CVP}}}}
\def\SAPROT{{\text{\sf{SA\text{-}PROT}}}}
\def\PROT{\mathsf{PROT}}
\def\op{\mathop{\mathrm{op}}\nolimits}
\def\Ops{\mathbf{Ops}}

\def\PSPACE{{\mathbf{PSPACE}}}
\def\CFL{{\mathsf{CFL}}}
\def\C{\mathbf{C}}

\def\ceps{\cancel{\eps}}

\def\wreps{{\eps\textbf{-write}}}
\def\wrceps{{\ceps\textbf{-write}}}

\title{On Computational Complexity of Set Automata\thanks{ The study
    has been funded by the Russian Academic Excellence Project
    '5-100'. Supported in     part by RFBR grants 16--01--00362 and
    17--01--00300.}}

\author{A. Rubtsov\inst{1}\inst{2}
\and
M. Vyalyi\inst{3}\inst{1}\inst{2}}
\institute{
 National Research University Higher School of Economics\\
\and Moscow Institute of Physics and Technology\\
\and Dorodnicyn Computing Centre, FRC CSC RAS\\
\email{\{rubtsov99, vyalyi\}@gmail.com}
}

\def\tp{\textbf{test+}}
\def\test{\textbf{test}}
\def\tm{\textbf{test-}}
\def\ins{\textbf{in}}
\def\out{\textbf{out}}
\def\wr{\textbf{write}}
\def\beg{\textbf{begin}}
\def\ends{\textbf{end}}

\usepackage{mathrsfs}

\def\CL{{\mathscr L}}

\usepackage[normalem]{ulem}
\usepackage{color}
\usepackage{xcolor}

\usepackage{enumerate}

\usepackage{marginnote}

\definecolor{WildStrawberry}{HTML}{FF43A4}
\definecolor{DarkGray}{HTML}{666666}

\usepackage{ifthen}

\newcommand{\malert}[2][]{\marginpar{\color{WildStrawberry}\ifthenelse{\equal{#1}{}}{\vspace{-2ex}}{\vspace{-2ex}\vspace{#1}}  #2}}

\newcommand{\mcomment}[2][]{\marginpar{\ifthenelse{\equal{#1}{}}{\vspace{-2ex}}{\vspace{-2ex}\vspace{#1}} \color{DarkGray} #2}}
\newcommand\reduline{\bgroup\markoverwith
      {\textcolor{red}{\rule[-0.5ex]{2pt}{0.4pt}}}\ULon}

\usepackage{dashrule}

\newcommand\alertuline{\bgroup\markoverwith
      {\color{WildStrawberry}{\hdashrule[-0.5ex]{2pt}{0.4pt}{1pt}}}\ULon}

\newcommand\commentuline{\bgroup\markoverwith
			      {\color{DarkGray}{\hdashrule[-0.5ex]{2pt}{0.4pt}{1pt}}}\ULon}

	\usepackage{mathtools} 
    \usepackage[makeroom]{cancel}
	\def\ceps{\cancel{\eps}}

\begin{document}

\maketitle

\begin{abstract}
We consider a computational model which is known as set automata.

The set automata are one-way finite automata
 with an additional storage---the set.  
		There are two kinds of set automata---the deterministic and the nondeterministic ones. We denote them as DSA and NSA respectively. The model was introduced by M. Kutrib, A. Malcher, M. Wendlandt in 2014 in \cite{KutribDLT2014} and \cite{KutribDCFS2014}. 
It was shown that
 DSA-languages look similar to DCFL due to their closure properties  and NSA-languages look similar to CFL due to their undecidability
properties.

		In this paper, which is an extended version of the conference paper  \cite{RV17}, we show that this similarity is natural: we prove that languages recognizable by NSA form a rational cone, so as CFL.
		
The main topic of this paper is computational complexity: we prove
that 
\begin{itemize}
\item [--] languages recognizable by DSA belong to $\PP$ and there are
$\PP$-complete languages among them;
\item [--]  languages recognizable
by NSA are in $\NP$ and there are
$\NP$-complete languages among them;
\item [--] the word membership problem is $\PP$-complete for DSA
  without $\eps$-loops and $\PSPACE$-complete for general DSA;
\item [--] the emptiness problem is in $\PSPACE$ for NSA and,
  moreover, it is $\PSPACE$-complete for DSA.
\end{itemize}
\end{abstract} 
	  
\section{Introduction}\label{Intro}	              
We consider a computational model which is known as set automata. A
set automaton is a one-way finite automaton equipped with an
additional storage---the set ${\mathbb S}$---which is accessible through the
work tape. On processing of a word, the set automaton can write a
word~$z$ on the work tape and perform one of the following
operations: the operation $\ins$ inserts the word $z$ into the set
${\mathbb S}$, the operation $\out$ removes the word $z$ from the set ${\mathbb S}$ if
${\mathbb S}$ contains $z$, and the operation $\test$ checks whether $z$
belongs to ${\mathbb S}$.  After each operation the work tape is
erased.

There are two kinds of set automata---the deterministic and the nondeterministic ones. We denote them as DSA and NSA respectively. 

If determinism or nondeterminism of an automaton is not significant, we use abbreviation SA, and we refer to the class of languages recognizable by (N)SA as $\SA$. We denote as $\DSA$ the class of languages recognizable by DSA.

\subsection{The Definition, Known Properties and Examples}

We start with formal definitions. A set automaton $M$
 is defined by a tuple
$$M = \langle S, \Sigma, \Gamma, \vartriangleleft, \delta, s_0, F \rangle, \text{where}$$
	\begin{itemize}
	\item $S$ is the finite set of states;
	\item $\Sigma$ is the finite alphabet of the input tape;
	\item	$\Gamma$  is the finite alphabet of the work tape;
	\item $\vartriangleleft \not\in \Sigma$ is the right endmarker;
	\item $s_0 \in S$ is the initial state;
	\item $F \subseteq S$ is the set of accepting states;
	\item $\delta $  is the transition relation:
\end{itemize}
$$\delta \subseteq S \times (\Sigma \cup \{\varepsilon, \vartriangleleft\}) \times \left[S \times (\Gamma^*\cup \{\textbf{in}, \textbf{out}\}) \cup S \times \{ \textbf{test} \} \times S \right]. $$

In the deterministic case $\delta$ is the function
$$\delta : S \times (\Sigma \cup \{\varepsilon, \vartriangleleft\}) \to \left[S \times (\Gamma^*\cup \{\textbf{in}, \textbf{out}\}) \cup S \times \{ \textbf{test} \} \times S \right]. $$
                                                                                                      
As usual, if $\delta(s, \eps)$ is defined, then $\delta(s,a)$ is not defined 
for every $a \in \Sigma$.

A \emph{configuration} of $M$ is a tuple $ (s,v,z, {\mathbb S})$ consisting of the state $s\in S$, the unprocessed part of the input tape $v\in\Sigma^*$, the content of the work tape $z\in\Gamma^*$, and the content of the set ${\mathbb S}\subset \Gamma^*$.  The transition relation determines the action of $M$ on configurations. We use $\vdash$ notation for this action. It is defined as follows
\begin{align}
	(s, xv, z, {\mathbb S} ) &\vdash (s^\prime, v, zz^\prime, {\mathbb S})  &&\text{if } (s, x, (s^\prime,z^\prime))\in\delta,
&&z'\in\Gamma^*;\\  
  (s, xv, z, {\mathbb S} ) &\vdash (s^\prime, v, \varepsilon, {\mathbb S}\cup\{z\}) &&\text{if } (s, x,(s^\prime,\textbf{in}))\in\delta;\label{in-query}\\ 
  (s, xv, z, {\mathbb S} ) &\vdash (s^\prime, v, \varepsilon, {\mathbb S}\setminus\{z\}) &&\text{if } (s, x,(s^\prime,\textbf{out}))\in\delta;\label{out-query}\\
  (s, xv, z, {\mathbb S} ) &\vdash (s_{+}, v, \varepsilon, {\mathbb S}) &&\text{if } (s, x,(s_+,\textbf{test},s_{-}))\in\delta,&& z \in {\mathbb S};\label{testp-query}\\
  (s, xv, z, {\mathbb S} ) &\vdash (s_{-}, v, \varepsilon, {\mathbb S}) &&\text{if } (s, x, (s_+,\textbf{test},s_{-}))\in\delta,&&z \not\in {\mathbb S}.\label{testm-query}  
\end{align}
Operations with the set (transitions~(\ref{in-query}--\ref{testm-query}) above) are called \emph{query operations}. A~word $z$ in a configuration to which a query is applied (the content of the work tape) is called a \emph{query word}.

We call a configuration \emph{accepting} if the state of the configuration is  accepting (belongs to the set $F$) and the word is processed till the endmarker. 
 So the accepting configuration
has the form $(s_f, \varepsilon, z, {\mathbb S})$, where $s_f\in F$.

The~set automaton
 accepts a word $w$ if there exists a run from the initial configuration  $(q_0, w\vartriangleleft, \varepsilon, \varnothing)$ to some accepting one.

  Set automata were introduced by M.~Kutrib, A.~Malcher, M.~Wendlandt
  in 2014 in \cite{KutribDLT2014} and \cite{KutribDCFS2014}. The
  results of these conference papers are covered by the journal
  paper~\cite{KutribSApaper2016}, so in the sequel we give references to the journal variant.

We recall briefly  results from~\cite{KutribSApaper2016} about  structural and decidability properties of $\DSA$. They are presented in the  tables, see Fig.~\ref{SAtables}.
 In the first table we list decidability problems: emptiness, regularity, equality to a regular language and finiteness.  
In  the tables, $R$ denotes an arbitrary regular language.
 The second table describes the structural properties: $L$, $L_1$ and $L_2$ are languages from the corresponding classes; we write $+$ 
in a cell if the class is closed under the operation, otherwise we write $-$.

	{
		\setlength{\tabcolsep}{7pt}     
		\renewcommand{\arraystretch}{1.2} 
		\setlength\columnsep{1em}
		\begin{figure}[t]
		   \begin{multicols}{2}

		   	 \begin{tabular}{|c|c|c|c|} 
					\hline
					 \phantom{!} & $\DSA$ & $\CFL$ & $\DCFL$ \\ \hline
					$L \stackrel{?}{=} \es $  & + & + & + \\ \hline
					$L \stackrel{?}{\in} \REG$   & \textbf{+} & $-$ & + \\ \hline
					$L \stackrel{?}{=} R $    & \textbf{+} & $-$ & + \\ \hline
					$|L| \stackrel{?}{<} \infty$    & \textbf{+} & $+$ & + \\ \hline
				\end{tabular}  
				\smallskip

				\begin{tabular}{|c|c|c|c|}
				 \hline 	  
				  &  $\DSA$ & $\CFL$ & $\DCFL$ \\\hline
				$L_1 \cdot L_2 $ 	  & $-$ & $+$ & $-$ \\\hline
				$L_1 \cup L_2$	   	  & $-$ & $+$ & $-$ \\\hline
				$L_1 \cap L_2$	   	  & $-$ & $-$ & $-$ \\\hline
				$\Sigma^*\setminus L$ & $+$ & $-$ & $+$  \\\hline
				  $L \cup R$	   	  & $+$ & $+$ & $+$ \\\hline
					$L \cap R$	   	  & $+$ & $+$ & $+$ \\\hline
				\end{tabular}   	
		   \end{multicols}  		   		  
			
			\caption{Structural and decidability properties }\label{SAtables}
	\end{figure}}

From Fig.~\ref{SAtables} one can see that $\DSA$ languages look similar to $\DCFL$. Let us consider an example of a DSA-recognizable language that is not a DCFL (and not even a CFL).

\begin{example}\label{PerExample}
	  Denote $\Sigma_k = \{0, 1, \ldots k-1\}$. We define $\Per_k = \{ (w\#)^n\,|\, w \in \Sigma_k^*, n \in \NN  \} $ to be the language of repetitions of words over $\Sigma_k$ 
separated by  the delimiter~$\#$.
For any $k$
 there exists DSA $M$ recognizing $\Per_k$.
\end{example}
\begin{proof}
 Firstly $M$ copies the letters from $\Sigma_k$ on the work tape until it meets $\#$. Then it performs 
the operation $\ins$
 on $\#$.
So, after processing of the prefix $w\#$,  the set ${\mathbb S}$
contains~$w$. 
Then, $M$ copies letters from $\Sigma_k$ on the work tape and performs 
the operation $\test$
 on each symbol~$\#$ until reaches~$\rendmarker$.
 DSA $M$ accepts the input iff 
all tests are positive; an accepted input
looks like $w\#w\#\ldots w\# \in \Per_k$. 
\qed\end{proof}

One can naturally assume that $\DSA$ class contains $\DCFL$ (or even $\CFL$), but this assumption is false. As was shown in~\cite{KutribSApaper2016}, 
the language
 $\{ w\#w^R \mid  w \in \Sigma^*, |\Sigma| \geq 2 \}$ is not recognizable by any DSA.
\begin{theorem}[\cite{KutribSApaper2016}]\label{DCFLDSA}
		The classes $\DCFL$ and $\DSA$ are incomparable. 
\end{theorem}

The undecidability results for NSA
 are based on the fact that NSA can accept the set of invalid computations of a Turing machine.
\begin{theorem}[\cite{KutribSApaper2016}]\label{NSAundecidability}
For NSA the questions of universality, equivalence with regular sets, equivalence, inclusion, and regularity are not 
semi-decidable. Furthermore, it is not semi-decidable whether the language accepted by some NSA belongs to $\DSA$.
\end{theorem}

We shall also mention a beautiful result from~\cite{KutribSApaper2016} about DSA languages over unary alphabet. It was shown that if language $L \subseteq a^*$ is recognizable by a DSA then $L$ is a regular language. Next example shows that for NSA it is false.

\begin{example}\label{PrimesExample}
  Let $\Primes = \{ a^p \mid \text{ $p$ is a prime number }  \}$. There exists NSA~$M$ with the unary alphabet of 
the work tape
recognizing language $\{a\}^*\setminus\Primes$ and this language is not regular. 
\end{example}

	\begin{proof}
   NSA~$M$ guesses a divisor $k$ of $n$ on an input word $a^n$, processes the input word $a^n$ letter by letter and copies letters to the work tape. After first $k$ letters $M$ performs $\ins$ operation, so ${\mathbb S}=\{a^k\}$. After each next $k$ letters (the position is guessed nondeterministically)  $M$ performs $\test$. $M$ accepts $a^n$ iff all the tests are positive and the last test was on the last letter. 
If $M$ guesses $k$ wrong at some point, then either some test after this guess is negative or there is no test on the last letter. In both cases $M$ doesn't accept $a^n$ on this run.
So, $a^n \in L(M)$ iff $n = km$. It is well-known that $\Primes$ is not a regular language as well as its compliment.
\qed	\end{proof}

It is worth to mention a quite similar model  presented by K.-J.~Lange and K.~Reinhardt in~\cite{Lange96setautomata}. We refer to this model as L-R-SA. In this model there are no $\ins$ and $\out$ operations; in the case of $\tm$ result the tested word is added to the set after the query; also L-R-SA have no $\eps$-moves. The results from~\cite{Lange96setautomata} on computational complexity for L-R-SA are  similar to ours: the membership problem is $\PP$-complete for L-R-DSA, and $\NP$-complete for L-R-NSA.

\subsection{Rational transductions}
Despite the fact that  the classes $\DCFL$ and $\DSA$ are incomparable, their structural and decidability properties are similar. In this paper we show that this similarity is natural: we prove that the class of languages recognizable by NSA has the same structure as context-free languages.

More exactly, both classes are  principal rational cones.

Let us recall the definitions.

	A \emph{finite state transducer} (FST) is a nondeterministic finite automaton with 
the output tape.
 Let $T$ be a FST. 
The set $T(u) $ consists of all words $v$ that $T$ outputs on runs from the initial state to a final state  while processing of $u$.
 So, FST $T$ defines a \emph{rational transduction} $T(\cdot)$. We also define $T(L) = \bigcup\limits_{u\in L}T(u)$.

The \emph{rational dominance} relation $A \lerat B$ holds
if there exists a FST $T$ such that $A = T(B)$, here $A$ and $B$ are languages.

A \emph{rational cone} is a family of languages $\C$
 closed under the rational dominance relation: 
$A \lerat B $ and $B \in \C$ imply $A \in \C$.
If there exists a language $F\in \C$ such that 
$L \lerat F$ 
for any $L \in \C$,
then $\C$ is a \emph{principal} rational cone generated by $F$; we denote it as $\C = \T(F)$.

Rational transductions for context-free languages were thoroughly investigated in the 70s, particularly by the 
French school. The main results of this research are published in J.~Berstel's book~\cite{BerstelBook}. As described in~\cite{BerstelBook}, 
it follows from the Chomsky-Sch\"utzenberger theorem
that $\CFL$ is a principal rational cone:   $\CFL = \T(D_2)$, where $D_2$ is the Dyck language on two brackets.

\subsection{Our contribution}

In this paper we address  to  computational complexity of problems related to SA.

In Section~\ref{complete-languages} we present examples of $\PP$-hard languages recognizable by  DSA and examples of $\NP$-hard languages recognizable by  NSA. Also, in this section, we put the word membership problem for DSA into hierarchy of complexity classes. 
We prove that the word membership problem is $\PP$-complete for DSA without $\eps$-loops. As corollary, we get the inclusion $\DSA \subseteq \PP$.
Also we prove that the word membership problem is $\PSPACE$-complete for general DSA (with $\eps$-loops).

For the rest of results we need a characterization of $\SA$ as a
principal rational cone generated by the language $\SAPROT$ of correct
protocols for operations with the set. It is a direct generalization
of well-known characterization of $\CFL$ as a principal rational cone generated by the Dyck language  $D_2$, which is in fact the language of correct protocols for operations with the stack (the push-down memory). This characterization is provided in Section~\ref{SAcones}.

It was shown in~\cite{KutribSApaper2016}
 that the emptiness problem for DSA is decidable. In fact the proof doesn't depend on determinism of SA. In Section~\ref{EmptinessSection} we  prove that
the emptiness problem for NSA is in $\PSPACE$ and the problem  is  $\PSPACE$-hard for DSA with the unary alphabet of the work tape. 
Our proof relies on the technique of rational cones.

The main technical result of the paper is $\SA \subseteq \NP$ (see Section~\ref{sect:SAinNP}). This result is based on the fact that class~$\SA$ is a rational cone and on our improvement of the technique of normal forms described in~\cite{KutribSApaper2016}.

\section{$\PP$ and $\NP$-complete languages}\label{complete-languages}

\subsection{$\PP$-complete languages recognizable by a DSA}

\begin{lemma}\label{P-Hard-Lemma}
 There exists a	$\PP$-complete language recognized by a DSA.
\end{lemma}

\begin{proof}
The proof idea is straightforward: we reduce a $\PP$-complete language to a~language recognizable by DSA.  

For the former language we adopt the language $\CVP$ (Circuit Value Problem), which is $\PP$-complete~\cite{Ladner:1975:CVP} under 
log-space reductions (we denote them by $\lelog$).
The variant of the language $\CVP$, which is convenient for our
purposes, consists of words that are encodings of assignment sequences
(CVP-programs).  Each variable $P_i$ is encoded by a binary string, the operation basis is $\land, \lor, \neg, 1, 0$.
An assignment has one of the following forms: $P_i := P_j \op P_k$,
where $\op \in \{\land, \lor\} $, $j,k<i$; $P_i := 1$; $P_i := 0$;
$P_i := \neg P_j$, where $j<i$. An assignment sequence belongs to
$\CVP$ iff the value assigned to the last variable is equal to one. We call it the \emph{value} of CVP-program.

We reduce $\CVP$ to  the language $\SACVP$.
It consists of words that are encodings of \emph{reversed assignment} sequences. A reversed assignment has one of the following forms:  $P_j \op P_k =: P_i$, where $\op \in \{\land, \lor\} $, $j,k<i$; $1=: P_i $; $0=: P_i$; $ \neg P_k =: P_i $, where $k<i$.  
Initially each variable is assigned to zero.
Unlike $\CVP$, reassignments of variables are allowed. A word $w$ belongs to $\SACVP$ if the last assignment (the value of the reversed $\CVP$-program) is equal to one. So, a word from $\SACVP$ looks like $$w = \#\langle P_1\rangle \#\land\#\langle P_2\rangle \#\langle P_3\rangle \# 1 \# \langle P_4\rangle \# \dots \#\langle P_j\rangle \#\op\#\langle P_k\rangle \#\langle P_i\rangle \#.$$

Now we prove that  $\SACVP$ is  recognized by a DSA $M$. 
  Note that  codes of the reversed assignments form a regular language. Thus 	we assume w.l.o.g. that the DSA $M$ processes  words, which are the sequences of reversed assignments $P_j \op P_k =: P_i$.    
Other words are rejected by verifying a regular event. 

	In the case of one-assignment $1 =: P_i$ the set automaton $M$ adds the code of $P_i$ to the set, in the case of zero-assignment $0 =: P_j$ the automaton 
removes  the code of $P_i$ from the set. 
In other cases the set automaton determines the value of operation by testing the codes of variables to be members of the set. $M$ stores the results of tests and computes the result of operation $P_j \op P_k$ in its finite memory. Then $M$ copies the code of $P_i$ on the work tape and performs operation $\ins$ if the computed result is $1$ and performs $\out$ operation otherwise. 

	The set automaton $M$ also stores in its finite memory the result of the last assignment and accepts if this result is $1$ when reaches the endmarker. It is clear that the set automaton $M$ accepts a reversed $\CVP$-program iff the value of the program is $1$.

 To complete the proof we note that the construction of a reversed $\CVP$-program  from a regular one is log-space computable. It gives a reduction  $\CVP\lelog \SACVP$. Thus, $\SACVP$ is $\PP$-hard. It is clear that $\SACVP$ in $\PP$.
\qed\end{proof}

\subsection{$\NP$-complete languages recognizable by a NSA}

Now we present  an NSA recognizing  an $\NP$-complete language.
 This result was also proved independently by M. Kutrib, A. Malcher, M. Wendlandt (private communication with M. Kutrib).
We construct an $\NP$-complete language
 that we call $\SASAT$ and reduce $\SAT$ to this language. The language $\SASAT$ is  quite similar to the language in~\cite{Lange96setautomata} 
for the corresponding result for L-R-SA model.

Let words $x_i \in \{0,1\}^*$ encode variables and 
words $\langle\varphi\rangle\in \{0,1\}^*$ encode 3-CNFs. $\SASAT$ contains the words of the form $x_1\#x_2\#\ldots\#x_n\#\#\langle\varphi\rangle $ such that an auxiliary 3-CNF $\varphi^\prime$ is satisfiable. The 3-CNF $\varphi^\prime$ is derived from $\varphi$ as follows.
For each variable $x_i$ that appears in the list $x_1\#x_2\#\ldots\#x_n\#\#$ at least twice, remove 
all clauses
 containing $x_i$ from $\varphi$ and 
get the reduced 3-CNF $\varphi''$. 
Set each variable $x$ of $\varphi''$ that is not in the list  $x_1\#x_2\#\ldots\#x_n\#\#$  
to zero, simplify $\varphi''$ and obtain as a result the 3-CNF~$\varphi^\prime$.

\begin{lemma}
	 The language $\SASAT$ is $\NP$-complete and it is recognized by an NSA.	
\end{lemma}
\begin{proof}
We describe
 an NSA $M$ that recognizes the language $\SASAT$. Firstly, $M$ processes the prefix $x_1\#x_2\#\ldots\#x_n\#\#$, guesses the values $b_i$ of $x_i$ ($0$ or $1$) and puts the pairs $(x_i, b_i)$ in the set. If a~variable $x_i$ appears more than once in the prefix, then the NSA $M$ nondeterministically guesses this event and puts both $(x_i, 0)$ and $(x_i, 1)$ in the set. On processing of the suffix 
$\langle\varphi\rangle$, 
for every clause of $\varphi$
 the NSA $M$ nondeterministically guesses a literal that satisfies the clause and verifies whether the set contains the required value of 
the corresponding variable.
 It is easy to see that all tests are satisfied along a run of $M$ iff 3-CNF $\varphi^\prime$ is satisfiable.

 The language $\SAT$ is reduced to $\SASAT$ in a straightforward way. Also it is easy to see that $\SASAT \in \NP$.
\qed\end{proof}

\subsection	{The membership problem for DSA}

An instance of the membership problem for DSA is a word $w$ and a
description of a DSA $M$.  The question is to decide whether $w\in
L(M)$.

As usually,  $\eps$-loops cause difficulties in  analysis of deterministic models with $\eps$-transitions. 
We say that DSA $M$ \emph{has $\eps$-loop} if there is a chain of $\eps$-transitions from some state $q$ to itself. 

At first, we show that DSA without  $\eps$-loops recognize all languages from the class $\DSA$.

Hereinafter, we call 
SA $M$ and $M'$ \emph{equivalent} if $L(M)=L(M')$.

\begin{proposition}\label{DSAepsLoops}
	For each DSA $M$ there is an equivalent DSA $M^\prime$ without $\eps$-loops.
\end{proposition}
\begin{proof}
  There are two kinds of $\eps$-loops: during the loops of the first kind $M$ only writes words on the work tape and during the loops of the second kind $M$ performs queries. 
The former are evidently useless, so we simply remove them. The latter are significant: the behavior of $M$ on these loops depends on the set's content and sometimes $M$ goes to 
an infinite loop
and sometimes not. 
 
Since $M$ is a deterministic SA, the set of words $\{u_1, \ldots, u_m\}$ that $M$ writes on all $\eps$-paths, is finite (after removing $\eps$-loops of the first kind). We build DSA $M^\prime$ by $M$ as follows. Each state of $M^\prime$ is marked by a vector $\vec a = (a_1, \ldots, a_m)$ of zeroes and ones. A content of the set 
is compatible with a vector~$\vec a$ if $a_i = 1$ is equivalent to $u_i \in {\mathbb S}$.
 Each state of $M^\prime$ has the form $\langle s, \vec a, \mathsf{aux}\rangle $, where $s$ is a state of $M$ and $\mathsf{aux}$ is an auxiliary part of finite memory that $M^\prime$ uses to maintain $\vec a$ correctly.

Non-$\eps$-transitions of $M'$ do not change components $\vec a$, $\mathsf{aux}$ of a state and change the first component according the transition function of $M$.

To define an $\eps$-transition~$\delta(\langle s, \vec a,  \mathsf{aux}\rangle, \eps )$ of $M^\prime$ 
we follow an $\eps$-path of $M$
 from the state $s$ 
with the set content compatible with $\vec a$. If $M$ goes to an  $\eps$-loop,
then the $\eps$-transition of $M^\prime$ 
is undefined. If the $\eps$-path finishes at a state $s^\prime$,
 then $\delta(\langle s, \vec a,  \mathsf{aux}\rangle, \eps ) = (\langle s^\prime, \vec a^\prime,  \mathsf{aux}^\prime\rangle, z )$, where $\vec a^\prime$ is the vector compatible with the $M$'s set content and $z$ is the content of $M$'s work tape at the end of the $\eps$-path. 

The transition function defined in this way may have  $\eps$-loops but only of the first kind (no queries along a loop). Deleting these loops we get the required DSA~$M'$.
\qed\end{proof}

\begin{theorem}\label{membership-without-eps}
  The membership problem for DSA without $\eps$-loops is $\PP$-complete. 
\end{theorem}

\begin{proof}
	An instance of the membership problem is an encoding of a~DSA $M$ without $\eps$-loops 
and a word $w$. 
 It is easy to emulate a DSA by a Turing machine with 3 tapes. 

The first tape is for the input tape of the DSA, the second one is for the work tape and the third one (the storage tape) is used to maintain the set.

Suppose $M$ processes a subword $u$ of the input between two queries to the set.
There are no $\eps$-loops. Therefore during this processing the set automaton writes at most $c|u|$ symbols on the work tape, where $c$ is a  constant depending only on the description of the set automaton.
It implies that $c|w|$ space on the storage tape is sufficient  to maintain the set content on processing of the input word~$w$.
Thus each query to the set can be performed in polynomial time in the input size.

So, the membership problem for DSA without $\eps$-loops belongs to $\PP$ and due to Lemma~\ref{P-Hard-Lemma} 
it is  $\PP$-complete.
\qed\end{proof}

From these results we conclude that all languages recognizable by DSA are in the class $\PP$ of polynomial time.

\begin{corollary}
 $\DSA \subseteq \PP$.  
\end{corollary}
\begin{proof}
Proposition~\ref{DSAepsLoops} guarantees that any language $L$ recognizable by a DSA is also recognizable by a DSA without $\eps$-loops. From Theorem~\ref{membership-without-eps} we conclude that $L\in\PP$.
\end{proof}

In the case of general DSA the membership problem is much harder. It appears to be $\PSPACE$-hard even in the case of testing whether the empty word $\eps$ belongs to the language recognizable by a DSA with the unary alphabet of
  the work tape. 

We call this restricted version  of the membership problem 
the $\eps$-\emph{membership problem}. An instance of the problem is a description of a DSA
$M$ with the unary alphabet of the work tape. The question is to check  $\eps\stackrel{?}{\in} L(M)$.

\begin{theorem}\label{empty-membership-DSA-unary}
  The $\eps$-membership problem is $\PSPACE$-hard. 
\end{theorem}

To prove Theorem~\ref{empty-membership-DSA-unary} we simulate
operation of a deterministic Turing machine $M$ with the 2-element alphabet
$\{0,1\}$ in polynomially bounded space by a DSA $A_M$ with the unary
alphabet $\{\#\}$ of the work tape. We describe only
$\eps$-transitions of DSA, because the rest of transitions does not affect the answer to
the question $\eps\stackrel{?}{\in} L(M)$.

Denote by $N$ the number of tape cells available  for operation of $M$. Let $Q$ be the
state set of $M$, $\delta\colon \{0,1\}\times Q \to  \{0,1\}\times Q
\times\{-1,0,+1\}$ be the transition function of $M$.

Our goal is to construct a DSA $A_M$ with $\poly(N)$ states simulating the operation of~$M$. A~state of $A_M$ carries an information about a position of the head of $M$ and a current line of the transition
function~$\delta$. Information about the tape content of $M$ is stored in the set  ${\mathbb S}$ as follows.

Let $i$ be the index of a tape cell of $M$, $1\leq i\leq N$, and
$x_i\in\{0,1\}$ be the value stored in this cell. In simulation
process the DSA $A_M$ maintains the relation: $x_i=1$ is equivalent to
$p_i\in{\mathbb S} $, where $p_i=\#^i$.

To insert words $p_i$ and delete them from the set, the states of the DSA $A_M$ include a counter to write the appropriate number of $\#$ on the
work tape. Thus the value of the counter is no more than~$N$.

As it mentioned above, the head position of $M$ and the current line
of the transition function~$\delta$ are maintained in the state set of
$A_M$. Thus, the set state of $A_M$ has the form
$[N]\times\Delta\times (\{0\}\cup[N])\times C$, where $[N]=\{1,2,\dots,N\}$, $\Delta =
\{0,1\}\times Q$ is the set of the lines of the transition function~$\delta$ and
$C$ is used to control computation, $|C|= O(1)$ (in the asymptotic
notation we assume that $N\to\infty$).

The DSA $A_M$  starts a simulation step in the state
$(k,(a,q),0,\start)$, where $k$ is the head position  of the simulated
TM, $q$ is its state, and $x_k = a$. Let $\delta(a,q) = (a',q',d)$ be
the corresponding line of the transition function~$\delta$.

A simulation step consists of the actions specified below, each action
is marked in the last state component to avoid ambiguity:
\begin{enumerate}
\item\label{write-action} write $k$ symbols $\#$ on the work tape, in this action the
  third state component  increased by 1 up to $k$;
\item update the set, if $a'=0$, then the DSA performs $\out $
  operation, otherwise it performs $\ins$ operation;
\item change the value of $k$ according the value of $d$: the new
  value $k'=k+d$;
\item write $k'$ symbols $\#$ on the work tape, the action is similar
  to the action~\ref{write-action};
\item perform the test operation, if the result is positive, then
  change the second component to $(1,q')$, otherwise change the second
  component to $(0,q')$;
\item change the third  component to $0$ and the last component to $\start$.
\end{enumerate}

Inspecting these actions, it is easy to see that all information about
the configuration of the simulated TM is updated correctly.

\begin{proof}[of Theorem \ref{empty-membership-DSA-unary}]
  The following problem is obviously $\PSPACE$-complete: an input is a
   pair $(M,1^N)$, one should decide whether
   $\eps\in L(M)$, here $M$ is an input TM with
   2-element alphabet that uses at most~$N$ cells on the tape.

   We reduce the problem to the question $\eps\stackrel{?}{\in}
   L(A_{M,N})$, where $A_{M,N}$ is the DSA with the unary alphabet of
   the work tape simulating operation of $M$ on $2N$ memory cells as
   described above. The initial state of $A_{M,N}$ is $(N, (0,q_0),
   0,\start)$, where $q_0$ is the initial state of $M$. 
 The
accepting states of $A_{M,N}$ are $(k,(a,q_f),0,\start)$, where $q_f$ is an
accepting state of the simulated machine $M$.

It is easy to see that the DSA  $A_{M,N}$ can be constructed in time
polynomial in the length of description of TM $M$ and $N$.
\qed\end{proof}

To get the matching upper bound of computational complexity of the membership problem, we refer to the complexity results about the emptiness problem that are proved below in Section~\ref{EmptinessSection}.

\begin{proposition}\label{member->emptiness}
  The membership problem for DSA is polynomially reduced
  to  the emptiness problem for DSA. 
\end{proposition}
\begin{proof}
  If a DSA $M$ and a word $w$ are given, then 
it is easy to construct in polynomial time a DSA $M'$ recognizing the
language $\{w\}\setminus L(M)$. It gives the required reduction, since $w\in L(M)$ iff $L(M')=\es$.
\end{proof}

Theorem~\ref{NSA-Emptiness-complexity} below places the emptiness problem for general SA into the class $\PSPACE$. Proposition~\ref{member->emptiness} implies that  the membership problem for DSA
is also in $\PSPACE$.

 \section{Structural Properties of SA-languages}\label{SAcones}

In this section we show that the class $\SA$ is a principal rational
cone generated by the language of correct protocols.

To simplify arguments, hereinafter we assume that an NSA  satisfies the following requirements:
\begin{enumerate}[(i)]
	\item the alphabet of the work tape is the binary alphabet,
	say, $\Gamma = \{a,b\}$;
	\item 
the NSA doesn't use the endmarker $\rendmarker$ on the input tape
and the initial configuration is $(q_0, w, \varepsilon,
\varnothing)$;
	\item the NSA accepts a word only if the last transition was an operation. 
\end{enumerate}

	\begin{proposition}\label{BinarySet}
	For any  SA $M$ there exists an equivalent SA $M^\prime$ satisfying the requirements \textup{(i)--(iii)}.
	\end{proposition}

	\begin{proof}
	  Let $\Gamma_M$ be the alphabet of the work tape of $M$.
	   Fix an enumeration of letters of $\Gamma_M$. Construct a SA $M''$ emulating $M$ and having the alphabet $\Gamma$ of the work tape such that, when $M$ writes $i$-th letter on its work tape, $M''$ writes the word $ba^{i}b$ on its work tape. It is clear that for each run of $M$ on $w$ there exists a corresponding run of $M''$ on $w$ such that all tests results are the same.

To complete the proof one need to transform $M''$ to satisfy the requirements (ii)--(iii). It is quite easy due to nondeterministic nature of SA. The modified SA $M'$ nondeterministically guesses the end of the input and stops operation just before $M''$ reads the endmarker (adding a dummy query if needed).
\qed	\end{proof}

	\begin{remark}
	  It is easy to see that the previous construction gives also a log-space algorithm that converts an NSA $M$ to the NSA $M'$. So, in algorithmic problems such as the emptiness problem below, one may assume that an input NSA satisfies the Requirements~(i)--(iii).
	\end{remark}

Now we introduce the main tool of our analysis: \emph{protocols}.
	
A protocol is a word $p = \# u_1\#{\op}_1\# u_2\#{\op}_2\#\cdots\#u_n\#{\op}_n$, where $u_i \in \Gamma^*$ and $\op_i\in \Ops=\{\ins,\out,\tp,\tm\}$. 
\emph{Query words} are words over the alphabet $\Gamma$.

We say that $p$ is \emph{a correct protocol for SA~$M$ on an input $w
 \in L(M)$}, if there exists a run of $M$ on the input $w$ such that
 $M$ performs the operation $\op_1$ with the query word $u_1$ on the work
 tape at first, then performs ${\op}_2$ with $u_2$ on the work tape,
 and so on.  In the case of a test operation, the symbol ${\op}_i$ indicates the result of
 the test: $\tp$ or $\tm$.

We call $p$ \emph{a correct protocol for SA}~$M$ if there exists a word $w\in L(M)$ such that $p$ is a correct protocol for SA~$M$ on the input $w$. And finally, we say that $p$ is a \emph{correct protocol} if there exists an SA~$M$ such that $p$ is a correct protocol for~$M$.
	
We define~$\SAPROT$ to be the language of all correct protocols over the alphabet of the work tape $\Gamma = \{a,b\}$.

	\begin{proposition}\label{SAPROTprop}
The language $\SAPROT$ is recognizable by DSA.
	\end{proposition} 
	\begin{proof}
		Note that the language of all protocols is a regular language. So we assume w.l.o.g. that an input of the DSA $M_\PROT$ recognizing $\SAPROT$ is a protocol. We construct~$M_\PROT$ in a straightforward way. It copies each $u_i$ to the work tape and performs operation indicated by the symbol $\op_i$. If $\op_i$ is a test symbol, then $M_\PROT$ verifies if the result of the test performed is consistent with $\op_i$. $M_\PROT$ accepts a protocol iff all  test results are consistent with  the corresponding symbols of the protocol.

	Note that the protocol of $M_\PROT$ on an accepted input $w$ is $w$ by itself. So each protocol accepted by $M_\PROT$ is a correct protocol by the definition. And it is obvious that $M_\PROT$ accepts all correct protocols.
\qed	\end{proof}

We are going to  prove that the class $\SA$ is a principal rational cone generated by the language $\SAPROT$. For this purpose we need a notion of \emph{inverse transducer}. 

We use a notation $q \xrightarrow[v]{u} p$ to express the fact that a transducer $T$ has a run from the state $q$ to the state $p$ such that $T$ reads
$u$ on the input tape and writes $v$ on the output tape. The notation is also applied to  a single transition. 
A rational transduction $T(u)$ mentioned  in the introduction is defined with this notation as $\{ v \mid q_0 \xrightarrow[v]{u} q_f, q_f \in F \}$.
 Recall that $T(L) = \bigcup\limits_{u\in L} T(u)$.

We define $T^{-1}(y) = \{ x \mid y\in T(x)\}$ and  $T^{-1}(A) = \{x \mid T(x)\cap A\ne\es\}$.
It is  well-known (e.g., see~\cite{BerstelBook})
that 
for every FST $T$
 there exists a FST $T^\prime$ such that $T^\prime(u) = T^{-1}(u)$. This transducer $T'$ is called  an \emph{inverse transducer}. 

Our goal is to prove that 
for every SA-recognizable language $L$
there exists a FST $T$ such that $L = T(\SAPROT)$. 

The plan of the proof is to construct  for an SA~$M$  a FST $T_M$ such that $w \in L(M)$ iff $T_M(w) \cap \SAPROT \neq \es$. We call such FST $T_M$ an \emph{extractor} (of a protocol) for $M$. Then we will show that the required FST $T$ is $T_M^{-1}$.

\begin{lemma}\label{ExtractorLemma}
For any SA~$M = \langle S, \Sigma, \Gamma, \vartriangleleft, \delta_{M}, s_0, F \rangle$
 there exists an extractor $T_M = \langle S \cup \{s^\prime_0\}, \Sigma, \Gamma, \delta, s^\prime_0, F \rangle$.
\end{lemma}

Let us informally describe the proof idea. 
The behavior of the extractor is similar to the behavior  of the SA.  When the SA writes something on the work tape, the FST writes the same on the output tape. When the SA makes a query, the FST writes a word $\#\op\#$. The only difference is that the SA knows the results of the performed tests. But a nondeterministic extractor can guess the results of the tests to produce a correct output.

\begin{proof}
As for transducers, we use a notation $s_1 \xrightarrow[y]{x} s_2$ to express the fact that an SA during a run from the state $s_1$ to the state $s_2$ reads the word $x$ from the input tape and writes the word $y$ to the work tape.

Let us formally define the transition relation $\delta$ of $T_M$.

There is an auxiliary transition $s_0^\prime \xrightarrow[\#]{\eps} s_0$ for the state $s_0^\prime$. 
If $M$ performs no query on a transition  $s_i \xrightarrow[y]{x} s_j$, $x \in \Sigma \cup \{\eps\}$, then this transition is included into the transitions of the extractor $T_M$.
If $M$ performs an operation $\op \in \{\ins, \out\}$ on the transition $s_i \xrightarrow{x} s_j$, then $T_M$ has the transition $s_i \xrightarrow[\#\op\#]{x} s_j$; in the case of a test, there are both transitions for $\op=\tp$ and $\op=\tm$ to the states $s_j^{+}$ and $s_j^{-}$ respectively.
 
The extractor $T_M$ also nondeterministically guesses the last query and doesn't write the last delimiter $\#$. At this moment it rejects an input iff the SA $M$ does not accept it.

Note that $T_M$ translates each input to  a protocol by the construction. If $w \in L(M)$, then each correct protocol for $w$ belongs to $T_M(w)$: $T_M$ guesses all  $M$'s tests results (recall that we assume the requirements (ii)--(iii) on SA). If $w \not\in L(M)$, then there is no sequence of test results that allows $T_M$ to write a~correct protocol on the output tape.
\qed\end{proof}

\begin{theorem}\label{ConeTheorem}
	 $\SA$ is a principal rational cone generated by  $\SAPROT$.
\end{theorem}
\begin{proof}
	 We shall prove that 
for every SA $M$
 there exists a FST $T$ such that $T(\SAPROT) = L(M)$. Take $T = T_M^{-1}$, where $T_M$ is the extractor for $M$. By the definition of extractor $w \in L(M)$ iff $T(w) \cap \SAPROT \neq \es$. So, if $w \in L(M)$, then there is at least one correct protocol in $T_M(w)$. Therefore
 $w \in T_M^{-1}(\SAPROT)$. In the other direction: if $w \in T_M^{-1}(\SAPROT)$, then $T_M(w) \cap \SAPROT \neq \es$. So, $w \in L(M)$ by the definition of extractor.

In the other direction, let $L = T(\SAPROT)$ be a language from the rational cone generated by $\SAPROT$. The language $L$ is recognized by an NSA $M_T$ which is a composition of the inversed transducer $T^{-1}$ and the DSA $M_\PROT$ recognizing correct protocols (see Proposition~\ref{SAPROTprop}). We build the composition of a~transducer and an~SA by a product construction. This composition is an NSA because the inversed transducer might be nondeterministic.
\qed\end{proof}

\section{Computational Complexity of the Emptiness Problem}\label{EmptinessSection}

Recall that 
an instance  of the emptiness problem for SA is a description of an SA~$M$. The question is
to decide whether $L(M)=\es$. 
The membership problem for DSA is reduced to the emptiness problem (Proposition~\ref{member->emptiness}). Thus, Theorem~\ref{empty-membership-DSA-unary} implies that the emptiness problem is $\PSPACE$-hard.

The main result of this section is the matching upper bound of computational complexity of the emptiness problem.

\begin{theorem}\label{NSA-Emptiness-complexity}
The emptiness problem for SA is in $\PSPACE$. 
\end{theorem}

It implies that the emptiness problem for SA is  $\PSPACE$-complete. More exactly,  
Theorems~\ref{empty-membership-DSA-unary} and~\ref{NSA-Emptiness-complexity} imply the following  corollaries.

\begin{corollary}
   The emptiness problem for DSA with the unary alphabet is
  $\PSPACE$-complete. 
\end{corollary}

\begin{corollary}
   The membership problem for DSA with the unary alphabet is
  $\PSPACE$-complete. 
\end{corollary}

\subsection{The emptiness problem and the regular realizability problem}

In this short subsection we present an application of the rational cones technique from Section~\ref{SAcones}. 
We show that the emptiness problem for 
NSA is equivalent to the regular realizability (NRR) problem for $\SAPROT$.

The problem $\nreg(F)$ for a language $F$ (a parameter of the problem) is to decide on the input nondeterministic finite automaton (NFA) $\A$  whether the intersection $L(\A) \cap F$ is nonempty.

\begin{lemma}\label{RRSAemptiness}
   $(L(M) \stackrel{?}{=} \es ) \lelog \nreg(\SAPROT) \lelog (L(M) \stackrel{?}{=} \es )$.
\end{lemma}
\begin{proof}
	It is easy to see that one can construct the extractor $T_M$ by SA~$M$ in log space. So, by the definition of extractor, we get that $L(M) = \es$ iff $T_M(\Sigma^*)\cap \SAPROT = \es $. Note that rational transductions preserve regularity~\cite{BerstelBook}: $T_M(\Sigma^*) \in \REG$. 
It is shown in~\cite{RVRR2015DCFS} that an NFA recognizing $T_M(\Sigma^*)$ is log-space constructible by the description of $T_M$. So $(L(M) \stackrel{?}{=} \es ) \lelog \nreg(\SAPROT)$.
	
It was shown in~\cite{KutribSApaper2016}  that 	 for any regular language $R \subseteq \Gamma^*$ there exists SA~$M_R$ recognizing $\SAPROT\cap R$. This SA  is also constructible in log space by the description of an NFA recognizing $R$: the proof is almost the same as for the aforementioned NFA by FST construction. Thus $\nreg(\SAPROT) \lelog (L(M) \stackrel{?}{=}~\es )$.
\qed\end{proof}

Due to Lemma~\ref{RRSAemptiness} to prove
Theorem~\ref{NSA-Emptiness-complexity} it is sufficient to prove the following fact. 

\begin{lemma}\label{NRRprotocolsInPSPACE}
$\nreg(\SAPROT)\in \PSPACE$.  
\end{lemma}

An idea of the proof  is
straightforward. We simulate a successful run of an automaton on a
protocol in nondeterministic polynomially bounded space. Applying
Savitch theorem we get  $\nreg(\SAPROT)\in \PSPACE$.

A possibility of such a simulation 
is based 
on a structural result about correct protocols for SA, see
Lemma~\ref{set-poly-bound} below. The next subsection introduces notions needed to this result and contains its proof.

We present the proof of  Lemma~\ref{NRRprotocolsInPSPACE} in Subsection~\ref{NRR-lemma-proof}.

\subsection{Successful automata runs on a protocol}

Let $\A$ be an NFA with the set of states $Q= Q_\A$ of size~$n$. Suppose that
there exists a~successful run of the automaton $\A$ on a protocol (a word
in the language $\SAPROT$). It implies the positive answer for the
instance~$\A$ of the
regular realizability problem $\nreg(\SAPROT)$. 
                  
Our goal in this section is to prove an existence of a successful run of
$\A$ on another correct protocol that satisfies specific requirements, see Lemma~\ref{set-poly-bound} below. 

To make exact statements we need to introduce a bunch of notation.

It is quite clear from the definition that a correct protocol just
describes a sequence of operations with the set.
Thus, a correct protocol
\begin{equation}\label{initial}
p = \# u_1\#{\op}_1\# u_2\#{\op}_2\#\cdots\#u_t\#{\op}_t
\end{equation}
determines the sequence $S_1$,
$S_2$, \dots, $S_t$ of the set contents: $S_i$ is the set content
after processing the prefix $\# u_1\#{\op}_1\#\cdots\#u_i\#{\op}_i$, i.e. performing  $i$ first operations with the set.

\emph{Query blocks}  (blocks in short) are the parts
of the protocol in the form $\# u_k\#\op_k$, $u_k\in\Gamma^*$, $\op\in\Ops$ (a~query word followed by an operation with the set).  Note that a query block is specified by a word $\# u_k\#\op_k$ and by a position of the word in the protocol: we distinguish different occurrences of the same word. 

This convention is useful to formulate   a criterion of protocol correctness. It can be expressed as follows. 

 We say that a query block $p_i$ \emph{supports} a query block $p_j$ if $u_i = u_j$ and $\op_i = \ins$, $\op_j = \tp$ or $\op_i = \out$, $\op_j = \tm$ and there is no query block $p_k$ such that $\op_k \in \{\ins, \out\}$, $u_k = u_i$ and $i < k < j$. Note that each $\tp$-block is supported by some $\ins$-block, but blocks with the operation $\tm$   may have no support in a correct protocol. \emph{Standalone blocks} are blocks in the form $p_j = \#u\#\op$, where $\op\in \{\out,\tm\}$, that  have no support (if $\op=\tm$)  and there is no block $p_k= \# u \#\ins$, where $k < j$.

The following lemma immediately follows from the definitions.

\begin{lemma}\label{ProtocolCorrectnessLemma}
  A protocol $p$ is correct iff each query block $\#u_i\#\tp$ is supported and each query block $\#u_i\#\tm$ 
is either supported or standalone.
\end{lemma}

Let $p$ be a correct protocol accepted by $\A$. Fix some successful run of $\A$ on $p$. The run is partitioned  as follows
\begin{equation}\label{SArunPartition}
s_0 \xrightarrow{\#u_1\#\op_1} s_1 \xrightarrow{\#u_2\#\op_2} s_2 \xrightarrow{\#u_3\#\op_3} \cdots \xrightarrow{\#u_{t}\#\op_{t}} s_t, \quad s_t\in F .	
\end{equation}
Here we assume that $\A$ starts from the initial state $s_0$, processes a query block~$p_i = \#u_i\#\op_i$ from the state $s_i$ and comes to the state $s_{i+1}$ at the end of the processing.
We say that~\eqref{SArunPartition} is the \emph{partition} of the run and two runs have the same partition if they have the same sequences of states $\{s_i\}$ and operations $\{\op_i\}$.

Let $\al$ be a triple $(q,q^\prime,\op)$, where $q,q^\prime\in Q$ and
$\op$ is an operation with the set, i.e.  $\op\in \mathtt{\Ops}$.
The language $R(\al)\in \Gamma^*$ consists of all words $u\in\Gamma^*$ such that
\[
q^\prime \in \delta_\A(q, \# u \# \op) .
\]

It is obvious that $R(\al)$ is a regular language and the minimal size of an
automaton recognizing it does not exceed~$n$ (the number of $\A$'s states). The total number $N$ of
languages $R(\al)$ is $\poly(n)$ (actually, $O(n^2)$). Each partition of the run~\eqref{SArunPartition} determines the sequence $\al_i = (s_i,s_{i+1}, \op_i)$; we say that
the partition has a type $\bal=(\al_i)_{i=1}^t$.

Suppose that  $\bal$ is a partition type of a
correct protocol. Generally, there are many partitions of correct
protocols having the type
$\bal$. We are going to choose among them as simple as possible. More
exactly, we are going to minimize the maximal size of the set contents
determined by correct protocols admitting a~partition of the type $\bal$.

For this purpose we start from a partition~\eqref{SArunPartition} of a
correct protocol and change some query words $u_i$ in it to keep the
partition type and the correctness of the protocol and to make the set contents smaller. 

Actually, we will analyze a slightly more general settings. We replace the family of regular  languages  $R(\al)$ indexed by triples $\al=(q,q^\prime,\op)$ by 
  an arbitrary finite family $\R= \{R(1),\dots,
R(N)\}$ of arbitrary languages over the alphabet $\Gamma$ indexed by
  integers $1\leq \al\leq N $; we call them \emph{query languages}. 

We extend the definition of type from partitions to protocols in a straightforward way:
we say that  a~protocol $p$~\eqref{initial} 
has a type $\bal = (\al_1,\dots,\al_t)$, $1\leq\al_i\leq N $, if $u_i\in R(\al_i)$ for each~$i$.

We transform the protocol~\eqref{initial} in two steps to achieve the
desired `simple' protocol. 

The first transformation of a correct  protocol~\eqref{initial}  of
type $\bal$ gets a
correct protocol 
\begin{equation}\label{1st-transform}
p' = \# u'_1\#{\op}_1\#u'_{2}\cdots  \#u'_t\#{\op}_t
\end{equation}
of the same type $\bal$
such that all set contents $S'_k$ determined by the
protocol~\eqref{1st-transform} are polynomially bounded in the number
$N$ of query languages. 

This property of a protocol will be used in our  $\PSPACE$-algorithm for the emptiness problem for SA. 

The formal definition of the transformation $p\to p'$ is somewhat
tricky. So we explain the intuition behind the construction. 

We are going to preserve  all operations with the set  in the transformed protocol. Also we are going to make a~sequence of the set contents $S'_k$ determined by the transformed protocol as monotone as possible. 

In general, it is impossible to make the whole sequence monotone. Thus we select a subset of query words (\emph{stable words} in the sequel) and do not change these words during the transformation. Thus, in the transformed protocol, all query blocks with stable query words  remain either supported or standalone (the initial protocol~\eqref{initial} is assumed to be correct). 

We shall get an upper bound $O(N^2)$ on the number of stable words from the formal definitions below.

The rest of query words are \emph{unstable words}. We are going to make the sequence of the set contents monotone on unstable words. 
It means that an unstable word added to the set is never deleted from it. And to satisfy this condition we are making the corresponding  $\{\out,\tm\}$-query
blocks  standalone, i.e.  we substitute unstable query words in these blocks by
words that are never presented in the set.

To satisfy this requirement for an  $\{\out,\tm\}$-block $p'_k$, the query language $R(\al_k)$ corresponding to  the block should be large enough.  Below we give a formal definition of large and small languages to satisfy this requirement.

Note that  to support a $\tp$-query
block $p'_k$ it is sufficient to have any word from the language $R(\al_k)$ in the set. The transformation strategy of $\ins$-blocks is to insert into the set the only one word from any large query language. In the formal construction below we call these words \emph{critical}. The number of critical words is at most the number of query languages~$N$. 

In this way  we come to a polynomial upper bound on the sizes of the set contents~$S'_k$.

Now we present formal definitions to realize  informal ideas explained above. 
We define the `small' languages by an iterative procedure. 
\begin{description}
\item [Step 0.] All query languages $R(i)$ that contain at most $N$ words are declared small.
\item [Step \boldmath$j+1$.] Let $W_j$ be 
the union of all query  languages that were declared  small on steps  $0,\dots,{j}$. All query languages $R(i)$ that contain at most $N$ words from $\Gamma^*\setminus W_j$ are declared small.
\end{description}
Query languages that are not small are declared  \emph{large}.

	Let $s$ be the last step on which some language was declared  small. Thus $W_s$ is the union of all small languages. The query words from $W_s$ are called \emph{stable}. It is clear from the definition that each unstable query word belongs to a large language.

It is easy to observe that there are relatively few stable query words.

\begin{proposition}\label{SmallWordsProp}
	$|W_s| \leq N^2$.
\end{proposition}
\begin{proof}
    The total number of small languages doesn't exceed $N$. Each small language contributes  to the set $W_s$ at most $N$ words.
\qed\end{proof}

To define critical query words we assume that the protocol~\eqref{initial} is correct and has the type $\bal=(\al_1,\dots, \al_t)$. Let $S_1,\dots, S_t$ be the sequence of the set contents determined by the protocol. 

An unstable query word $u$  is \emph{critical} if it is contained in an $\ins$-block  $p_k =\# u\#\ins$ indicating  insertion into the set an unstable word $u$ from a large language, whose unstable words has not been  inserted into the set earlier. Formally it means that  there exists a large language $R(i)$
 such that $u_k\in R(i)\setminus W_s$ and $(R(i)\setminus W_s)\cap S_j=\es$ for all $j<k$.

\begin{proposition}\label{crit-upbnd}
  There are at most $N$ critical query words in the correct protocol~$p$~\eqref{initial}.
\end{proposition}
\begin{proof}
  The protocol $p$ is correct. Thus the first occurrence  of a critical query word $u$ is in an $\ins$-block $p_k=\# u\# \ins$. It means that $S_k\cap (R(i)\setminus W_s)\ne\es$ for all large languages $R(i)$ that contain $u_k$. 

  Therefore a large language can certify the critical property for at most one critical query word.  But the number of large languages is at most $N$.
\qed\end{proof}

Look at a $\{\tp,\ins\}$-block $p_k=\# u_k\#\op_k$ containing an unstable noncritical query word $u_k$. Note that 
 the query language $R(\al_k)$, which is specified by an index $\al_k$, $1\leq \al_k\leq N$, from the type of the protocol~\eqref{initial},
 is large (otherwise the query word $u_k$ would be stable). The query word $u_k$ is not critical. Thus $(R(\al_k)\setminus W_s)\cap S_j\ne\es$ for some $j<k$. The smallest $j$ satisfying this condition indicates the query block describing the first insertion of an unstable query word $u_j=\tilde u_k $ from the language $R(\al_k)$ into the set. Thus $(R(\al_k)\setminus W_s)\cap S_i=\es$ for $i<j$. It means that the query word $\tilde u_k $ is critical.  We assign the query word $\tilde u_k$ to the block $p_k$ and use this assignment in the construction below.

We are ready to give  exact definition of the transformation $p\to p'$ of protocols assuming correctness of $p$.

The type of the  transformed protocol is the same: $\bal=(\al_1,\dots, \al_t)$. Operations do not change: $\op'_i=\op_i$. 

All stable query words  do not change. More exactly, if  $u_k\in W_s$, then $u'_k=u_k$. 

Critical query words in $\{\tp,\ins\}$-blocks do not change: if $u_k$ is critical and $\op_k\in \{\tp,\ins\}$, then $u'_k= u_k$.

An unstable noncritical  query word  in a $\{\tp,\ins\}$-block $\# u_k\#\op_k$ is substituted by the critical query word $\tilde u_k$ assigned to the block.

An unstable query word in a $\{\tm,\out\}$-block  $\#u_k\#\op_k$ is substituted by a word $u'_k$ from $R(\al_k)\setminus (C\cup W_s)$, where $C$ is the set of critical query words. The substitution is possible because there are more than $N$ words in $R(\al_k)\setminus W_s$ (the language $R(\al_k)$ is large) and there are at most $N$ critical query words due to Proposition~\ref{crit-upbnd}.

\begin{proposition}\label{small-set-content}
  If the  protocol $p$~\eqref{initial} is correct, then the transformed  protocol~$p'$ \eqref{1st-transform} is correct.

  Each set content $S'_k$ determined by the  protocol $p'$ contains at most  $N^2+N$  words.
\end{proposition}
\begin{proof}
The set content $S'_k$ determined by the protocol $p'$ can contain only two kinds of query words: stable
query words (there are at most $N^2$ of them due to
Proposition~\ref{SmallWordsProp}) and  critical query words
 (there are at most $N$ of them due to Proposition~\ref{crit-upbnd}). It
gives the required upper bound on the size of the set contents.

Stable query words in both protocols occur in the same blocks. Since the protocol $p$ is correct, all test blocks of the protocol~$p'$ containing stable query words are either supported or standalone.

An unstable query word in a $\tm$-block of the protocol~$p'$ is noncritical by definition  the construction. Thus the block is standalone.

It follows from the correctness of $p$ that  the
first occurrence  of a critical query word $u$ in the protocol $p$ is in an $\ins$-block. 
Queries in  $\out$-blocks of the transformed protocol $p'$ do not contain critical  query words. Therefore, 
the query word $u$ is in the set at any later moment. It implies that each $\tp$-block containing the query word $u$ is supported in the protocol~$p'$.
\qed\end{proof}

Query words in a correct protocol may be very long. To operate with them in polynomial space, we describe them implicitly 
in the $\PSPACE$-algorithm below. All relevant information about a word is a list of query languages containing it.
 Thus,  we  divide the whole set of words
over the alphabet~$\Gamma$ in the non-intersecting \emph{elementary
  languages}
\begin{equation}\label{elementary-language}
  R_I = \bigcap_{i\in S} R(i) \cap \bigcap_{i\notin S} \overline{R(i)}, 
  \quad I\subseteq [N].
\end{equation}
We call $I$ a \emph{type} of a query word $u$, if $u\in R_I$. Words will be represented by their types in the algorithm below.

Such a representation causes a problem: the set content is  represented in this way by types of words and it is unclear how many different words of the same type are in the set. 

To avoid this problem, we assume that the only one word of each type is in the set at any moment. To justify the assumption, we need the second transformation of protocols.

To define the second transformation of the protocol~\eqref{initial},
let choose two query words~$u_I$, $v_I$ in each elementary language
$R_I$ containing at least two words. 
The transformed protocol
\begin{equation}\label{2nd-transform}
  p''= \# u''_1\#{\op}_1\# u''_2\#{\op}_2\#\cdots\#u''_t\#{\op}_t
\end{equation}
is produced as follows. 

The type of $p''$ equals the type of $p$. All operations are the same. 

If a word $u_i$ belongs to some language $R_I$ that contains only one word, we do not change $u_i$. Otherwise we substitute $u_i$ by $u_I$ in the case of $\op_i\in\{\ins, \tp\}$ or by $v_I$ in the case of $\op_i\in\{\out, \tm\}$. Here $I$ is the type of $u$.

\begin{proposition}\label{R_I-unique}
  If the initial protocol~\eqref{initial} is correct, then the transformed protocol~\eqref{2nd-transform} is also correct. 

  Each set content $S''_k$ determined by the transformed protocol contains at most one word from any elementary language $R_I$, $I\subseteq[N]$.
\end{proposition}
\begin{proof}
In the modified protocol each $\ins$-block with a query word of a  type $I$ contains the same query word~$u_I$ provided $|R_I|\geq 2$. So, the only word from the elementary language $R_I$
that can be in the set is~$u_I$. 

To prove correctness of the transformed protocol we use Lemma~\ref{ProtocolCorrectnessLemma}.  Query blocks with query words from 1-element elementary languages do not change. The initial protocol~\eqref{initial} is correct. So each query block of this form is either supported or standalone.
A~block $\#v_I\#\tm$, where $|R_I|\geq2$, is standalone because the word $v_I$ is never added to the set while performing operations described by the protocol~\eqref{2nd-transform}. 

In the initial protocol any $\tp$-block $p_k = \#u_k\#\tp$, where $u_k\in R_I$ and $|R_I|\geq2$, is supported. 
Thus at least one preceding block $p_i = \#u_i\# \ins$, $i<k$ contains a query word $u_i$  of the  type~$I$. 
It implies that  
the block $p'_k = \#u_I\#\tp$ of the protocol~\eqref{1st-transform} is supported by the block $p'_j=  \#u_I\#\ins$, $i\leq j<k$.
\qed\end{proof}

Now we return to successful runs of an automaton on a correct protocol. Propositions~\ref{small-set-content} and~\ref{R_I-unique} immediately imply the following lemma.

\begin{lemma}\label{set-poly-bound}
  Let $\A$ be an NFA with $n$ states. If $\A$ accepts
  a~correct protocol, then it accepts a correct protocol such that $|S_i|=
  \poly(n)$ for all~$i$ and  at most one word from each elementary language can belong to the set. 
\end{lemma}
\begin{proof}
Just apply the second transformation to the result of the first transformation. The second transformation do not increase the sizes of the set contents.
\end{proof}

\subsection{Proof of Lemma~\ref{NRRprotocolsInPSPACE}}\label{NRR-lemma-proof}

Let $\A$ be an input automaton for the emptiness problem. We measure
the size of the input by the number of states of this automaton.

As it mentioned above, we simulate a successful run of the automaton on
a protocol by a nondeterministic algorithm using polynomially bounded
space.

Due to Lemma~\ref{set-poly-bound} there exists a protocol $p$ and a
successful run of the automaton $\A$ on the protocol such that the
sizes of set
contents during the protocol are upper bounded  and the set contains at most one word from each elementary language. The upper bound is polynomial
in $n$ and can be derived explicitly from the arguments of the previous subsection.  It is easy to verify that $M =
32n^4$ upperbounds the sizes of set contents determined by the protocol from
 Lemma~\ref{set-poly-bound}.

The second condition  of Lemma~\ref{set-poly-bound} implies that to describe the set content it is sufficient to indicate elementary languages intersecting it. An elementary language can be described by a set $I$ of 
triples in the form $(q_1,q_2, \op)$,  where $q_1,q_2\in Q(\A)$ and
 $\op\in \{\ins,\out,\tp,\tm\}$. The language $R_I$ is specified by
Eq.~\eqref{elementary-language}. Such a representation has polynomial size, namely,  $O(n^2)$.

The simulating algorithm stores the description of the set content and
the current state of the automaton~$\A$. The algorithm
nondeterministically guesses the change of this data after reading the
next query block of the simulated protocol. 

To complete the proof we should devise an algorithm to check the
correctness of a simulation step. It  is specified by
indicating the state of $\A$ after reading a~query block, an elementary language
$R_I$, containing the query word on this step, where $I$ is the type of a query block to be read, and the description of the set content after performing the operation with the set on this step. 

It is easy to check that the state of $\A$ is changed correctly using the description of $\A$. 

The  language $R_I$ should be nonempty (otherwise the simulation goes wrong).
It is well-known that the intersection problem for regular languages
is in $\PSPACE$~\cite{Kozen:1977:LBN:1382431.1382559,Lange1992}. Thus, verifying $R_I\ne\es$ can be done in polynomial
space. It guarantees the possibility of the query indicated by the current query block. 

Test results are easily verified by the the description of the set content before a~query.

If the current query is $\ins$, then the set content is changed as follows. 
Due to Lemma~\ref{set-poly-bound} we assume  that there is at most one word of each type $I$ in the set. So, if the current query word has type~$I$ and there are  no words of type $I$ in the set, then the type $I$ is included in the set. Otherwise, the set is not changed.

If the current query is $\out$ and the set includes the current query type~$I$, then  two outcomes are possible: either a~word of the type $I$ is deleted from the set  or  the set  remains unchanged. The latter is possible iff the elementary language $R_I$ contains at least two words. 

Thus, to complete the proof we need to prove a proposition below.

\begin{proposition}\label{intersection>1}
  There exists a polynomial space algorithm to verify $$\big|\bigcap_{i}R_i\big|\geq2,$$ 
where $R_i$ are regular languages represented by NFA. 
\end{proposition}
\begin{proof}
  We reduce the problem to the intersection emptiness problem for regular languages, which is in $\PSPACE$ as it mentioned above. 

Let $R_1$, \dots, $R_m$ be regular languages over the alphabet $\Sigma$. We define languages $B, B_1, \dots, B_m$ over the alphabet $\Sigma \cup\{\lozenge\}$, where $\lozenge\notin \Sigma$, as follows.

The language $B_i$ consists of words $u_1v_1\dots u_tv_t$ such that $u_1u_2\dots u_t\in R_i\lozenge^*$ (a~word from $R_i$ padded by dummy symbols), 
$v_1v_2\dots v_t\in R_i\lozenge^*$. An automaton recognizing $B_i$ is constructed from an automaton recognizing $R_i$ in polynomial time by obvious modification.

The language $B$ consists of words $u_1v_1\dots u_tv_t$ such that $u_i\ne v_i$ for some $1\leq i\leq t$. It has a constant state complexity assuming the alphabet $\Sigma $ is fixed.

It is clear from the construction that
\[
\big|\bigcap_{i}R_i\big|\geq2\quad\text{iff}\quad
B\cap \bigcap_{i}B_i\ne\es.
\]
If $u\ne v$ are two different words in the $\bigcap_{i}R_i$, then they can be padded by dummy symbols $\lozenge$ to equal lengths and the perfect shuffle of the padded words is in $B\cap \bigcap_{i}B_i$. In other direction, any word from $B\cap \bigcap_{i}B_i$ produces two different words from $\bigcap_{i}R_i$ by taking symbols from odd and even positions respectively and deleting dummy symbols.
\qed\end{proof}

\begin{remark}
  In Statements~\ref{small-set-content} and~\ref{R_I-unique}  we do not restrict languages
  $R(i)$, $1\leq i\leq N$, to be regular. So, our proof smoothly extends on any   algorithmic problem in the from $\SAPROT\cap L(D)\stackrel{?}{=} \es $ for a
  model of computation $D$ such that both  intersection problems indicated in the above proof for
  $D$ are in $\PSPACE$. 

  This generalization of our theorem looks rather interesting and is
  worth to further investigation.
\end{remark}

\section{$\SA \subseteq \NP$}\label{sect:SAinNP}

We come to the main result of the paper. 
The general idea is very natural:  to prove that for each word from a~language recognizable by an NSA  there exists a~\emph{short} (polynomial) accepting run of the NSA  for the word. 
Thus an NP-algorithm guesses this run and verifies that the NSA accepts on that run.

We assume in asymptotics that the length of the input word $w$ tends to infinity. So we  call an object (a word, a protocol, a graph, etc) \emph{short} if the object description  is polynomial in $|w|$. Also, we call an object \emph{constant} if the object description is $O(1)$.

Realization of the above idea is technically hard. To implement it we provide a detailed analysis of successful runs of NSA and corresponding correct protocols and introduce a special form of NSA.
We still assume that Requirements~(i-iii) from Section~\ref{SAcones} hold for NSA.

Before starting a long series of definitions, we explain the main
obstacle to realize the above plan. Possible $\eps$-transitions in an
NSA run would make the total number of the set operations in the run (i.e. the number of query blocks in the corresponding protocol)  
arbitrary large.  If we had no $\eps$-transitions in NSA, we would
have had the same proof as for  $\text{L-R-NSA}\in \NP$
in~\cite{Lange96setautomata}. But $\eps$-transitions require more
careful consideration.

To overcome this difficulty we will use a special form of NSA that we call \emph{Atomic Action Normal Form} (AANF). AANF is a refinement of the infinite action normal form from~\cite{KutribSApaper2016} developed for DSA.

\subsection{Atomic Action Normal Form}

At first, we provide the definition of the action normal form (ANF) and then we define AANF as a special form of ANF.

\begin{definition}
	We say that SA $M$ is in the action normal form if the initial state appears only once on each run of $M$, and each state is either marked by operations $\mathbf{test+}$, $\mathbf{test-}$, $\mathbf{in}$, $\mathbf{out}$, if on a transition to this state $M$ performs the respective operation, or marked by $\wr$ in the other case (we assume that if SA doesn't perform a query, then it writes something on the work tape, maybe $\eps$).
\end{definition}

 So, if $M$ is in ANF, then its states are divided into groups according marks: $$S = \{s_0\} \cup S_{\tp} \cup S_{\tm} \cup S_{\ins} \cup S_{\out} \cup S_\wr,\quad S_{\op} \cap S_{\op^\prime} = \es,\quad \text{if }\op \neq \op^\prime.$$ 
Denote $$S_{\beg} = \{s_0\} \cup S_{\tp} \cup S_{\tm} \cup S_{\ins} \cup S_{\out},\quad S_\ends = S_\beg \setminus \{s_0\}$$ the sets of 
the start and the end states of queries.

\begin{lemma}[\cite{KutribSApaper2016}]\label{ANFLemma} For each SA $M$ there exists an equivalent SA $M^\prime$ in ANF.
\end{lemma}

In fact Lemma~\ref{ANFLemma} was stated in~\cite{KutribSApaper2016} for DSA. Its proof doesn't depend on determinism, but it preserves determinism of SA. Thus, from now on, we assume that NSA into consideration are in ANF.

Before introducing a more specific form of NSA, we need to introduce notions to describe properties of NSA runs and corresponding protocols, which will be important in the proof below.

In analysis of the emptiness problem we have used partitions of runs
(see Eq.~\eqref{SArunPartition}) of NFA that correspond to
protocols. Here we use similar partitions for NSA runs. But now we should
take into account symbols that are read from the input tape.

Any accepting run  of SA $M$ on an input word~$w $ is divided by queries into 
\emph{query runs}
\begin{equation}\label{SArun}
s_0 \xrightarrow[u_0]{x_0} s_1 \xrightarrow[u_1]{x_1} s_2 \xrightarrow[u_2]{x_2} \cdots \xrightarrow[u_{n-1}]{x_{n-1}} s_n, \quad s_n\in F,	
\end{equation}
where we assume that the SA $M$ starts from a state $s_i$ with the blank work tape, writes a query word $u_i$ on the work tape on processing of a word $x_i$ on the
input tape and performs a query at a state $s_{j+1}$.  
Thus $w= x_0x_1\cdots x_{n-1}$, $x_i \in \Sigma^*$ and the corresponding protocol has the form $$p = \# u_0\#\op_0\# u_1\#\op_1\dots\# u_{n-1}\#\op_{n-1} .$$ 

We will call partitions in the form~\eqref{SArun} \emph{run partitions}. A~protocol is determined by a run partition. Thus we will use for run partitions the notions introduced for protocols. In particular, query runs will be indicated by the corresponding query blocks. Note that the run partition~\eqref{SArun} assigns a~subword $x_i$ of the input word to the query block $\# u_i\#\op_i$. We call this word $x_i$ a \emph{query generator} for the block $\# u_i\#\op_i$. Full information about the query run corresponding to the query block $p_i$ of the protocol includes also the start state $s_i$ and the terminal state $s_{i+1}$.

\emph{Segments} of a partition are  sequences of query runs $$s_i\xrightarrow[u_i]{x_i} s_{i+1}
\xrightarrow[u_{i+1}]{x_{i+1}} \cdots \xrightarrow[u_{j-1}]{x_{j-1}} s_{j} .$$ A~segment may contain many queries. Thus we denote it as $s_i \xrightarrow{x_ix_{i+1}\cdots x_j} s_j$.

Note that 
$x_i=\eps$ is possible 
 and query runs of this form cause the main difficulty described above: 
a \emph{$\eps$-segment} $s_i \xrightarrow{\eps} s_j$ may contain a lot of 
queries that support the tests on the nonempty $x_k$s. On the other hand, there are only polynomially many query blocks with non-empty query generators. Thus we treat the cases of empty and non-empty generators in different ways.

Our goal is to transform an accepting run partition~\eqref{SArun} into a short accepting run partition. We call a transformation \emph{valid} if it preserves accepting run partitions. Correctness of the corresponding protocol is now not sufficient for a~transformation to be valid. We also need to maintain the relations between query words in blocks and their query generators. 

For this purpose we  use a relation that holds if $M$ can perform a query with a~query word $u$ after processing $x$ starting from $s$ and querying at $q$.  We denote it as $s \xhookrightarrow[u]{x} q$.
A~sister relation $s \xhookrightarrow[U]{X} q$, where $X$, $U$ are languages in the corresponding alphabets, means that $U = \{ u \mid \exists x \in X : s \xhookrightarrow[u]{x} q \}$.

We denote $L_{s_i,s_j}$ a language such that the relation $s_i \xhookrightarrow[L_{s_i,s_j}]{\Sigma^*} s_j $ holds, $s_i \in S_\beg$, $s_j \in S_\ends$. So, $L_{s_i,s_j}$ consists of all the words that $M$ can write on the work tape on a path from $s_i$ to $s_j$ on processing some word without performing queries in between.

Using this notation, we restate Lemma 8 from~\cite{KutribSApaper2016} in a way convenient for our needs.
\begin{lemma}[Lemma 8~\cite{KutribSApaper2016}]\label{ExcludeFiniteness}
	Let $F$ be a finite language. For SA $M$ in ANF there exists an equivalent SA $M^\prime$ in ANF such that for each pair of states $s_i^\prime \in S^\prime_\beg$, $s_j^\prime \in S^\prime_\ends$ $$L_{s_i^\prime, s_j^\prime} = L_{s_i, s_j}\setminus F$$ for some states  $s_i, s_j$ of $M$.
\end{lemma}

This lemma implies that each SA can be converted to an equivalent SA in \emph{infinite action normal form}: an ANF such that each language $L_{s_i,s_j}$ is either infinite or empty. But for our purposes we need more restrictive requirements on $L_{s_i, s_j}$.

Informally, we require that  states of a set automaton in this
restricted form distinguish between empty and non-empty query
generators.  Also, to simplify  replacements of query words in the run
partition we group query words into  equivalence classes. The exact
definition follows.

\begin{definition} We say that 
an NSA $M$ is in atomic action normal form (AANF) if the following conditions hold.
	\begin{itemize}
		\item Requirements \emph{(i-iii)} from Section~\ref{SAcones} hold.
		\item 	$S = \{s_0\} \cup S_{\tp} \cup S_{\tm} \cup S_{\ins} \cup S_{\out} \cup S_\wreps \cup S_\wrceps $, 
$S_{\op} \cap S_{\op^\prime} = \es$ if $\op \neq \op^\prime$:
 each state $q$
 such that the relation $s \xhookrightarrow[u]{x} q$ holds for some $x$ and $u$
is marked by the operation of the query; states that occur while writing $u$ are marked as $\wreps$ if $x = \eps$ and as $\wrceps$ if $x \neq \eps$.
		\item There exists a finite family of regular languages $\CL_M$ such that 
			\begin{itemize}
		    \item for all $A_\alpha, A_\beta \in \CL_M: A_\alpha \cap A_\beta = \es $ if $\alpha \neq \beta$;
				\item  if $s \xhookrightarrow[U]{\Sigma^*} q$, $s^\prime \xhookrightarrow[U^\prime]{\Sigma^*} q^\prime$, then 
either $ U \cap U^\prime = \es$ or $ U = U^\prime \in \CL_M$;

\item $|A_\alpha| = \infty$ for every $A_\alpha \in \CL_M$.
			\end{itemize}                                                                                           
	\end{itemize}
\end{definition}

AANF implies that 
each query run $s \xrightarrow[u]{x} q$
 belongs to an equivalence class corresponding to $A_\alpha \in \CL_M$, $u \in A_\alpha$, and the number of such classes is finite. Moreover, in the case of $x  = \eps$ either the relation $s \xhookrightarrow[A_\alpha]{\eps} q$ holds for some $\alpha$ 
or the relation  $s \xhookrightarrow[\es]{\eps} q$ holds (the state $q$ can not be reached from the state $s$ by $\eps$-transitions). The different languages $A_\alpha$ do not intersect. As it is seen below, these two conditions  make easier a construction of valid transformations  of run partitions.

\begin{lemma}\label{AANFlemma} 
For any NSA $M$
 there exists an equivalent NSA $M^\prime$ in AANF.
\end{lemma}

\begin{proof}  
   	Due to Lemma~\ref{ANFLemma} we assume that NSA $M$ is in ANF. At first we construct an auxiliary SA~$\tilde M$ that is equivalent to $M$.  We replace each state from $s_i \in S_\beg$ by states $(s_i,\eps)$, $(s_i,  \ceps)$, so if $M$ has a transition to $s_i$, then $\tilde M$ has corresponding transitions to both states $(s_i,\eps)$, $(s_i,  \ceps)$. We require $\tilde M$ to process only the empty word during each query that starts in $(s_i,\eps)$ and process at least one letter when it starts in $(s_i,  \ceps)$.	
	It is easy to implement these constraints as follows. We additionally mark $\wr$ states by $\eps$ and $\ceps$ and add an extra bit to the NSA state set. To satisfy the first constraint we  throw out all $\ceps$-moves for the states marked $\eps$. To satisfy the second constraint we use the extra bit  to verify that at least one letter have been processed before performing the query in a state marked $\eps$. Such transformation leaves $\tilde M$ in ANF.

		Let $\CL = \{ L_{s_i, s_j} \mid s_i \in S_\beg,  s_j \in S_\ends \text{ states of } \tilde M \}$. It is a family of constant size since the number of $\tilde M$'s states is constant (does not depend on an input word $w$). 

		 We take the closure of $\CL$ under $\cap$ and $\setminus$ operations  and obtain a finite family $\CL_1$. Next, we remove from $\CL_1$ all the languages that are the unions of some other languages and obtain a finite family $\CL_2$. By the construction, each language from $\CL$ is a finite union of some languages of $\CL_2$, but some languages from $\CL_2$ can be finite. Let $F$ be the union of all finite languages from $\CL_2$ and finally let $\CL_3$ consists of all infinite languages from $\CL_2$.

	Now we show how to transform $\tilde M$ to $M^\prime$ in such a way that $$\CL_{3} = \{ L_{s^\prime_i, s^\prime_j} \mid s^\prime_i \in S^\prime_\beg,  s^\prime_j \in S^\prime_\ends \text{ states of } M^\prime \}.$$ Note that $\CL$ is a family of regular languages and so is $\CL_3$. Denote $\CL_{3} = \{A_1, \ldots, A_\alpha, \ldots, A_m\}$.  At first, we transform $\tilde M$ to $M^{\prime\prime}$ by applying Lemma~\ref{ExcludeFiniteness} for $\tilde M$ and $F$. We get that $L_{s_i^{\prime\prime}, s_j^{\prime\prime}} = \bigcup\limits_{\alpha\in I} A_\alpha$, where $I \subseteq \{1,\ldots,m\} $.

Now we construct $M^\prime$ by $M^{\prime\prime}$. We replace a state $s_i^{\prime\prime}$ by states $(s_i^{\prime\prime}, \alpha)$ for each $\alpha \in I$ and demand that $M^\prime$, starting from $(s_i^{\prime\prime}, \alpha)$, writes on the work tape only the words from $A_\alpha$. We implement these constraints in the same way as at the beginning of the proof.

So, $\CL_{3} = \{ L_{s^\prime_i, s^\prime_j} \mid s^\prime_i \in S^\prime_\beg,  s^\prime_j \in S^\prime_\ends \text{ states of } M^\prime \}$ and all the constraints of AANF on $A_\alpha$s are satisfied by the construction and all the states are marked according to the definition of AANF.  
\qed\end{proof}

\subsection{The main part of the proof}

From now on we fix an NSA $M$ in AANF and the family of regular
languages $\CL_M = \{A_1, A_2, \ldots, A_m\}$ from the definition of
AANF. Recall that $|\CL_M| =m=O(1)$ (the family construction uses NSA, not
inputs of NSA).   We say that a
query block $p_i$ is an \emph{$\eps$-block} of a protocol if the
corresponding query run has the form $s_i \xrightarrow[u_i]{\eps}
s_{i+1} $ (the query generator of the block is empty).  We say that an
$\eps$-block $p_i$ has type $\alpha$ if the relation $s_{i}
\xhookrightarrow[A_\alpha]{\eps} s_{i+1}$ holds for $A_\alpha \in
\CL_M$.

Now we are ready to describe the main steps of the proof. Each step is a transformation 
of an accepting run partition to some other accepting run partition of $M$ on the same input~$w$ (a~valid transformation).

With compare to the protocol transformations used  in the analysis of the emptiness problem in Section~\ref{EmptinessSection}, valid transformation in this proof are more sophisticated. 

To define them we need to introduce  \emph{chains} of different kinds.

We say that query blocks $v_1, v_2, \ldots, v_k$ form a
\emph{chain}~$C$ if $v_1$ supports $v_i$, $2\leq i \leq k$.  Each
standalone query block forms a \emph{standalone chain}.  Let $v_i =
\#u\#\op_i$ be the blocks of a chain. We denote the chain by $C(u)$
and we call the word $u$ the \emph{pivot} of $C(u)$. Also we denote the chain
$C^{+}$ if $\op_1 = \ins$ and $C^{-}$ if $\op_1 \in \{\out, \tm\}$.

We call a chain $C(u)$ of $\eps$-blocks an \emph{$\eps$-chain}. We
denote $\eps$-chain $C_\alpha$ if all blocks in the chain have type
$\alpha$. Note that different blocks of an $\eps$-chain can't have different
types due to AANF definition.

This observation leads to  transformations that are defined as follows.
\begin{itemize}
\item Choose in 
 each $A_\alpha \in \CL_M$ a pair of words
 $u^{+}_\alpha \in A_\alpha$ and $u^{-}_\alpha \in A_\alpha$. 
\item Replace
  the query words $u$ in all blocks of $\eps$-chains $C^{+}_\alpha(u)$ and $C^{-}_\alpha(u)$
by $u^{+}_\alpha$ and $u^{-}_\alpha$ respectively.
\item Do not change the generators $x_i$ and the states $s_i$ in the partition.
\end{itemize}

Transformations of this kind will be called \emph{$\eps$-chain unifications}.
Note that any $\eps$-chain unification keeps the relations $s_i \xrightarrow[u^\pm_\al]{\eps} s_j$ in the transformed run partition.

\begin{lemma}\label{TwoWordsLemma}
 There exist short words $u^{+}_\alpha \in A_\alpha$ and $u^{-}_\alpha \in A_\alpha$ for any $A_\alpha \in \CL_M$
such that the $\eps$-chain unification using these pairs is a valid transformation 
of the run.
Moreover, during the transformed run
	\begin{itemize}
		\item  $M$ never adds 
the words $u^{-}_\alpha$
 to ${\mathbb S}$;
		\item $M$ never removes
the words $u^{+}_\alpha$
 from ${\mathbb S}$;
		\item all blocks of $\eps$-chains $C^{-}$ become standalone;
		\item the length 
of any word $u^{\pm}_\alpha$
 is $O(|w|)$.
	\end{itemize}
\end{lemma}
\begin{proof}
In each run~\eqref{SArun}
there is no more than $|w|$ nonempty $x_i$s and therefore no more than $|w|$ pivots $u_i$ of non-$\eps$-chains.

 Since each $A_\alpha\in\CL_M$ is an infinite regular language, it contains an infinite sequence of words, the lengths of which form an arithmetic progression. The common difference of the progression is $O(1)$ since $A_\alpha$ is constant (does not depend on the input word $w$). We choose words $u^{+}_\alpha \neq u^{-}_\alpha$ from this sequence such that both of them differs from any pivot $u_i$ of a non-$\eps$-chain.

Perform the  $\eps$-chain unification using the pairs $u^{+}_\alpha$, $ u^{-}_\alpha$. 

For $i$th block of the type $\al$
the relation  $s_i \xrightarrow[A_\alpha]{\eps} s_j$ holds, so the relation $s_i \xrightarrow[u^{\pm}_\alpha]{\eps} s_j$ also holds since $u^{\pm}_\alpha \in A_\alpha$. Thus the replacement $u_i$ by $u^{\pm}_\alpha$ in each $\eps$-block preserves the relations  $s_i \xhookrightarrow[A_\alpha]{\eps} s_{i+1} $  for $\eps$-blocks.  The rest of the blocks do not changed. So to prove that the transformation is valid we shall prove that the protocol is transformed to a correct one. We prove it 
using Lemma~\ref{ProtocolCorrectnessLemma}:
show that each query block is either supported, or a standalone one.
		
	 Note that, by the choice of $u^{\pm}_\alpha$, according to the transformed protocol the SA $M$ never adds a word $u^{-}_\alpha$ to ${\mathbb S}$ and never removes a word $u^{+}_\alpha$ from ${\mathbb S}$. Therefore all query blocks of $\eps$-chains $C_\alpha^{-}$ become standalone: $\tm$ queries by the words that never occur in ${\mathbb S}$ requires no support and $\out$-queries do not support any $\tm$-query anymore.
All query blocks of each $\eps$-chain $C^{+}_\alpha$ are remain supported since $u_\alpha^{+}$ is never removed from the set after the first addition. Each test result in an $\eps$-chain $C^{+}_\alpha$ is positive, since~$u_\alpha^{+} \in {\mathbb S}$.  
	   
	Since the transformation do not affect non-$\eps$-chains,
blocks in these chains remain either supported or standalone, as it
was before the transformation.  
\qed\end{proof}

From now on we assume w.l.o.g. that all run partitions into consideration
have the properties described in Lemma~\ref{TwoWordsLemma}. We denote
by  $U$  the set of words $u_\alpha^{\pm}$.  All words in this set are
short as it stated in Lemma~\ref{TwoWordsLemma}.

To prove the main result we need two more transformations. 

The first one makes the lengths of $\eps$-segments  in a partition
short. (The length of a segment is the number of query runs in it.) It does not affect query words in  non-$\eps$ query blocks.

More exactly, consider an $\eps$-segment $s_i \xrightarrow{\eps} s_j$
of a run partition. Recall that $$s_i \xrightarrow{\eps} s_j = s_i \xrightarrow[u_i]{\eps} s_{i+1} \xrightarrow[u_{i+1}]{\eps}\ldots \xrightarrow[u_{j-1}]{\eps} s_j.$$  We call the segment \emph{pure} if no query run $s_k \xrightarrow[u_i]{\eps} s_{k+1}$, $i\leq k<j$, in the segment supports a  non-$\eps$-chain.

\emph{$\eps$-shrinkage}  is  a transformation of a run partition such
that (i) it replaces pure $\eps$-segments by pure $\eps$-segments 
of constant length; (ii) it does not   affect the rest of the partition.

\begin{lemma}\label{EpsSegmentLemma}
If a run partition is obtained by an $\eps$-chain unification of an
accepting run partition, then  there exists a valid  $\eps$-shrinkage
of it. Moreover, in the transformed partition the lengths of all $\eps$-segments are short. 
\end{lemma}
\begin{proof}
  It is easy to show that any $\eps$-shrinkage makes the lengths of all $\eps$-segments short. 

   Fix the partition of an  $\eps$-segment $s_i\xrightarrow{\eps} s_j$ 
in the transformed run partition  in the form
	$$s_i \xrightarrow{\eps} s_{l_1} \xrightarrow[u_{l_1}]{\eps} s_{r_1} \xrightarrow{\eps} s_{l_2} \xrightarrow[u_{l_2}]{\eps} s_{r_2} \xrightarrow{\eps} \dots  s_{l_k} \xrightarrow[u_{l_k}]{\eps} s_{r_k} \xrightarrow{\eps}  s_j,$$ where  $s_{l_t} \xrightarrow[u_{l_t}]{\eps} s_{r_{t}} $ are all query runs  in the segment  that support   non-$\eps$-chains. So, all segments $s_{r_t} \xrightarrow{\eps} s_{l_{t+1}}$ are pure. There are at most $|w|$ query runs in the whole run partition having non-$\eps$ query generators. Thus there are at most $|w|$  non-$\eps$-chains and  $k\leq|w|$. 

Since all pure $\eps$-segments in the transformed run partition are constant, we conclude from $k\leq |w|$ that  the length of the segment $s_i\xrightarrow{\eps} s_j$ is~$O(|w|)$.

We define a valid $\eps$-shrinkage separately for each pure $\eps$-segment
\begin{equation}\label{pure-eps-seg}
  s_0 \xrightarrow{\eps} s_\ell = s_0 \xrightarrow[u_0]{\eps} s_{1} \xrightarrow[u_{1}]{\eps}\ldots \xrightarrow[u_{\ell-1}]{\eps} s_\ell\,.
\end{equation}
(For the sake of simplicity we change the indexes of states in the segment.)

Recall that  $\CL_M = \{A_1, A_2, \ldots, A_m\}$, $m=O(1)$, is the family of
  regular languages from the definition of AANF. 
  We mark each state $s_i$ of the segment~\eqref{pure-eps-seg} by a
  vector $\vec a = (a_1, \ldots, a_m)$ according to the content of ${\mathbb S}$ at this point of the run. A component $a_\alpha$ consists of two bits: the first one marks whether ${\mathbb S}$ contains a word from $A_\alpha\setminus U$ and the second marks whether $u^{+}_\alpha \in {\mathbb S}$:
	 \begin{itemize}
	   	 	\item  $a_\alpha = 00$ if $u^{+}_\alpha \not\in {\mathbb S}$ and $ {\mathbb S}\setminus U \cap A_\alpha = \es $;
			\item  $a_\alpha = 01$ if $u^{+}_\alpha  \in {\mathbb S}$ and $ {\mathbb S}\setminus U \cap A_\alpha = \es $;
			\item  $a_\alpha = 10$ if $u^{+}_\alpha  \not\in {\mathbb S}$ and $ {\mathbb S}\setminus U \cap A_\alpha \neq \es $;
			\item  $a_\alpha = 11$ if $u^{+}_\alpha  \in {\mathbb S}$ and $ {\mathbb S}\setminus U \cap A_\alpha \neq \es $.
	 \end{itemize}
Let $\vec a^k$ be the mark of $s_k$ in the
segment~\eqref{pure-eps-seg}. 
	
Note that query words in $\ins$, $\out$ and $\tm$ query runs of the
segment are the words from $U$ due to conditions of
Lemma~\ref{TwoWordsLemma}. But query words in $\tp$ query runs may be
out of~$U$ (these ones are in non-$\eps$-chains). Due to this reason
we need to keep two bits  in each component of vectors
$\vec a$.

It appears that vectors $\vec a$ bear enough information to produce
the transformed run. To build it we use an auxiliary
NFA $\A$ with the input alphabet $\{1,\dots,m\}$. The states of $\A$ are  
 $Q\times \{00,01,10,11\}^m = V $, where $Q = S_\beg \cup
S_\ends$. The initial state of $\A$ is $(s_0,\vec a^0)$, the only
accepting state is $(s_\ell,\vec a^\ell)$.

The transition relation $\delta_\A$ is defined by the rule:
$((q_l,\vec a),\al,(q_r,
\vec a^\prime)) \in\delta_\A$ iff  the relation $q_l \xhookrightarrow[A_\alpha]{\eps} q_r$ holds for
  the states $q_l$, $q_r$ and  
			\begin{itemize}

				\item $q_r \in S_\out \cup S_\tm $ and $\vec a = \vec a^\prime$;
				\item $q_r \in S_\tp $ and  $\vec a = \vec a^\prime$, $a_\alpha \neq 00$;
				\item $q_r \in S_\ins $ and $\vec a_\beta = \vec  a^\prime_\beta$, $\beta \neq \alpha$,   and $\vec a^\prime_\alpha = 01$ if $ \vec a_\alpha 
				\in \{00,01\} $; $\vec a^\prime_\alpha = 11$ if $ \vec a_\alpha 
				\in \{10,11\} $.
			\end{itemize}

It is easy to see from the definition that the
segment~\eqref{pure-eps-seg} generates a word
$\al_0\al_1\dots\al_{\ell-1}$, where $u_i\in A_{\al_i}$, 
accepted by $\A$. We are going to prove a relation in the opposite direction.

\textsc{Claim:} If $\al_0\al_1\dots\al_{t-1}$ is accepted by the NFA
$\A$, then there exists a segment
\begin{equation}\label{pure-eps-seg-transformed}
  s_0=s'_0 \xrightarrow[u'_0]{\eps} s'_{1} \xrightarrow[u'_{1}]{\eps}\ldots \xrightarrow[u'_{t-1}]{\eps} s'_{t}=s_\ell
\end{equation}
such that  $u'_i\in A_{\al_i}$ and 
a substitution of the
segment~\eqref{pure-eps-seg} by the
segment~\eqref{pure-eps-seg-transformed} gives an accepting run
partition.

The lemma easily follows from the Claim. Note that the NFA $\A$ is
constant (its construction does not depend on NFAs input) and $L(\A)\ne\es$. Therefore there exists a word
$\al_0\al_1\dots\al_{t-1}$ of the length $O(1)$ that is accepted
by~$\A$. Applying the Claim we obtain the valid transformation of the pure
segment~\eqref{pure-eps-seg} into a segment~\eqref{pure-eps-seg-transformed} of constant length.

\textsc{Proof of the Claim.} 
Note that there are no $\eps$-transitions in the
transition relation of $\A$.
Let 
\[
(s'_0,\vec a^0),\ (s'_1,\vec a^1),\ \dots,\ (s'_{t-1},\vec a^{t-1}),\ (s_\ell,\vec a^\ell)
\]
be the state sequence along an accepting run of NFA $\A$ on the input
$\al_0\al_1\dots\al_{\ell-1}$. These states is used to form the
segment~\eqref{pure-eps-seg-transformed}. 

Let ${\mathbb S}_{0}$ denotes the set content determined at the beginning of
the segment~\eqref{pure-eps-seg} of an accepting run partition.
Define query words $u'_i$ by the rules
\begin{itemize}
\item if $s'_{i+1}\in  S_\out \cup S_\tm $, then $u'_i= u^{-}_{\al_i}$;
\item if $s'_{i+1}\in  S_\ins $, then $u'_i= u^{+}_{\al_i}$;
\item if $s'_{i+1} \in S_\tp $ and $a^{i}_{\al_i}\ne 10$, then $u'_i=u^{+}_{\al_i}$;
\item if $s'_{i+1} \in S_\tp $ and $a^{i}_{\al_i}= 10$, then $u'_i$ is
  a word from $({\mathbb S}_{0}\setminus U)\cap A_{\al_i}$.
\end{itemize}

The last line should be justified. If $a^{i}_{\al_i}= 10$,
then $a^{0}_{\al_i}= 10$, 
since the transitions do not change the first bit in each
component of vectors $\vec a^{i}$
and the second bits in each components of vectors $\vec a^{i}$
form a nondecreasing sequence. But $\vec a^0$ corresponds to the
set content ${\mathbb S}_0$. Therefore $({\mathbb S}_{0}\setminus U)\cap A_{\al_i}\ne
\es$.

From the above rules it is easy to conclude that the transformed partition
is an accepting one. Indeed, $\out$-operations
on the segment do not change the set content. Thus query words in
$\tm$-queries  does not belong to the set and the test result is
negative as required. 

Note  that  the second bit in a vector component is increased only on
transitions to a state in $S_\ins$. In this case a word $u^+_{\al_i}$
is added to the set according to the rules. 

A~transition to the state in $S_\tp $ is possible iff the
 component $\vec a^i_{\al_i}$ is not $00$.

If  $\vec a^i_{\al_i}=10$, then a tested word belongs to
$({\mathbb S}_{0}\setminus U)\cap A_{\al_i}$. The test result is positive.

If  $\vec a^i_{\al_i}\in\{ 01,11\}$, then from the above observations
we conclude that either $a^{0}_{\al_i}=1$ and $u^+_{\al_i}$ is in the
set from the very beginning of the segment, or $a^{0}_{\al_i}=0$,
$a^{j}_{\al_i}=1$, $j<i$, and $u^+_{\al_i}$ is
added to the set on the $j$th operation with the set. In both cases
the test result is positive.

It completes the proof of the Claim.
\qed\end{proof}

Look at a run partition produced from an accepting run partition by
an $\eps$-chain unification followed by an $\eps$-shrinkage. The query
words in $\eps$-chains are short. The lengths of all $\eps$-segments are also short
but these segments may contain long query words in the blocks from
non-$\eps$-chains. The number of non-$\eps$-chains is $O(|w|)$.

Thus, to get a desired short accepting run we need a final touch: 
make query words in non-$\eps$-chains short.

\begin{lemma}\label{ChainReplacementLemma}
  Let  a partition
\begin{equation}\label{2ndSArun}
s_0 \xrightarrow[u_0]{x_0} s_1 \xrightarrow[u_1]{x_1} s_2 \xrightarrow[u_2]{x_2} \cdots \xrightarrow[u_{n-1}]{x_{n-1}} s_n, \quad s_n\in F	
\end{equation}
is produced from an accepting run partition by
an $\eps$-chain unification followed by an $\eps$-shrinkage. 

  Assume that $C(u)$ is a non-$\eps$-chain of query blocks $v_1, \ldots v_k$, $v_i = \#u\#\op_i$. Then there is a short word $u^\prime$ such that the replacement the word $u$ by the word $u'$ in the chain
is a~valid transformation of run partitions. 
\end{lemma}

\begin{proof}
A~replacement of query words in the chain $C(u)$ by a word $u^\prime$
is a valid transformation  iff relations $s_i \xhookrightarrow[u^\prime]{x_i} s_{i+1}$ hold for all blocks in the chain and the replacement of $u$ by $u^\prime$ does not break any other chain.  

It is sufficient to  satisfy relations $s_i
\xhookrightarrow[u^\prime]{x_i} s_{i+1}$ for non-$\eps$ blocks
only. Indeed, the query word  $u$ of the chain belongs to some language $ A_\al$ from the family $\CL_M$.  Due to this definition  the relation $s_i \xhookrightarrow[A_\alpha]{\Sigma^*} s_{i+1}$ holds for all blocks of the chain and the relation $s_i \xhookrightarrow[A_\alpha]{\eps} s_{i+1}$  holds for all $\eps$-blocks of the chain. 
Thus, if  $s_i \xhookrightarrow[u^\prime]{x_i} s_{i+1}$ holds for  $x_i\ne \eps$, then $u'\in A_\alpha$ and 
 the relation $s_i \xhookrightarrow[u^\prime]{\eps} s_{i+1}$ also holds for $x_i=\eps$.

To preserve all other chains, it is sufficient that a word
$u'$ does not belong to a short finite set $F$ of forbidden words. The forbidden words are $u^{\pm}_\alpha$ and all the words $u_i$ (except of~$u$) such that there is a chain $C(u_i)$ in the partition run~\eqref{2ndSArun}. Note that $|F|$ is $O(|w|)$ since the number of $u^{\pm}_\alpha$ is $O(1)$ by Lemma~\ref{TwoWordsLemma} and the number of all non-$\eps$-chains is $O(|w|)$.

	Denote by $R_i$ a regular language such that the relation $s_i
	\xhookrightarrow[R_i]{x_i } s_{i+1}$ holds, $x_i\neq\eps$. To
	prove the lemma it is sufficient to  prove that there is a
	short word in $\bigcap R_i \setminus F$.
	
Note that the number of $R_i$ is bounded by $|w|$, but it may grow with $|w|$. So if $R_i$ were arbitrary regular languages, the intersection $\bigcap R_i$ could contain no short words. 

To overcome this difficulty we exploit a specific form of languages $R_i$. Let $x_i = w_{l_i}w_{l_{i}+1}\cdots w_{r_i}, w_k \in \Sigma$. Let $q_k$ be the states such that during the run specified by~\eqref{2ndSArun} there is a transition from $q_k$ to $q_{k+1}$ processing the symbol $w_k  $ on the input tape; if $k = l_i$, then $q_k = s_i$; if $k+1 = r_i$, then $q_{k+1} = s_{i+1}$. There could be many choices for $q_k$ inside the block $s_i \xrightarrow{x_i} s_{i+1}$ -- there are no restrictions on that choice. Define $R(w_k)$ to be the language such that the relation  $q_k \xrightarrow[R(w_k)]{w_k} q_{k+1}$ determined by the extractor $T_M$ holds. In other words, $R(w_k)$ consists of all words that can be written on the work tape while NSA processes the symbol $w_k$ starting from $q_k$ and finishing at $q_{k+1}$ (all intermediate states should be in $S_\wr$). 
All these languages are regular and the number of them is constant (does not depend on an input $w$). So, we get that $R_i \supseteq R(w_{l_i})\cdot R(w_{l_i+1})\cdots R(w_{r_i}) = R^i_1\cdot R^i_2 \cdots R^i_{m_i}$, where $R^i_j = R(w_{l_i+j-1}),\, m_i = r_i-l_i+1$.
        
Since the number of $x_i \neq \eps$ does not exceed $|w|$, we get that $\sum_i m_i \leq |w| $. Note that $u \in \bigcap\limits_{i=1}^{m} R^i_1\cdot R^i_2\cdots R^i_{m_i} \setminus F$ by the definition of $R^i_j$.
It implies that for each $i$ there is a factorization of $u = v_1\cdot v_2\cdots v_{m_i}$ such that $v_j \in R^i_j$.  So, there exists a partition $u = \tilde u_1\tilde u_2\cdots \tilde u_K$, $\tilde u_k \in \Gamma^*$, such that all factors of the above factorizations have the form $\tilde u_l \tilde u_{l+1}\cdots \tilde u_{r} \in R^i_j$. It is easy to see that the length $K$ of the partition satisfies inequality $K \leq \sum_i m_i \leq |w|$.   
          
Let NFA $\A^i_j$ recognizes $R^i_j$. For each $\A^i_{j}$ there exists a successful run $$q_l \xrightarrow{\tilde u_l} q_{l+1} \xrightarrow{\tilde u_{l+1}} q_{l+2}\to \cdots \to q_k \xrightarrow{\tilde u_k} q_{k+1} \to \cdots   \xrightarrow{\tilde u_{r}} q_{r+1},$$ 
where $q_l$ is the initial state and $q_{r+1}$ is an accepting state of $\A^i_j$. Define NFA $\B^i_k$ as a modification of $\A^i_{j}$ by replacement  of the initial state to $q_k$ and the set of accepting states to $\{q_{k+1}\}$. Then $L(\A^i_{j}) \supseteq L(\B^i_{l})\cdot L(\B^i_{l+1})\cdots L(\B^i_{r})$ and
                  $$ R_i \supseteq L(\B^i_{1})\cdot L(\B^i_{2}) \cdots L(\B^i_{K}).$$

Note that for each $j$ the intersection $\bigcap\limits_i L(\B^i_j) = R^\prime_j$ is nonempty and contains~$\tilde u_j$.
  
 We prove that required $u^\prime$ exists and belongs to $R^\prime_1\cdot R^\prime_2 \cdots R^\prime_K $. If so, then $u^\prime$ evidently belongs to $\bigcap\limits_i R_i$, since $\bigcap\limits_i R_i \supseteq R^\prime_1\cdot R^\prime_2 \cdots R^\prime_K$ .

As was already mentioned for other objects, the number of all possible $\B^i_j$ is constant (does not  depend on $w$). The same is true for the state complexities of~$R^\prime_j$. If all the $R^\prime_j$ are finite languages, then the length of the longest word from $R^\prime_j$ is $O(1)$  and therefore $u^\prime = u$ is a short word: it belongs to $R^\prime_1\cdot R^\prime_2 \ldots R^\prime_K $ by the construction. If at least one language $R^\prime_j$ is infinite, then, as in the proof of Lemma~\ref{TwoWordsLemma}, it contains a sequence of words with linear length's growth. Fix all the shortest words  $u_k^\prime \in R^\prime_k$, $k \neq j$, and choose the first word $u^\prime_j \in R^\prime_j$ from this sequence such that $u^\prime = u^\prime_1\cdot u^\prime_2 \cdots u^\prime_j \cdots u^\prime_K \not\in F$. The word $u^\prime$ is short since $|F|$ is $O(|w|)$.
\qed\end{proof}

Tying up loose ends we get the main result.

\begin{theorem}\label{SAinNP}
 There is an NP-algorithm 
verifying for an input word $w$ whether $w \in L(M) $, where $M$ is an NSA.

\end{theorem}
\begin{proof}
W.l.o.g.  $M$ is NSA in AANF.
We will show that for any $w \in L(M)$  there exists a short accepting run. Thus, an NP-algorithm guesses a short run and checks its correctness. 

By definition, there is a accepting run for $w$.
At first, 
we apply to it an $\eps$-chain unification. By 
 Lemma~\ref{TwoWordsLemma} it gives an accepting run on the input $w$ such that all $\eps$-chains contain only short query words. 

Then we apply an $\eps$-shrinkage. By Lemma~\ref{EpsSegmentLemma} it gives an accepting run on the input $w$ such that all $\eps$-chains contain only short query words and all $\eps$-segments are short. 

After that we  transform all query words in non-$\eps$-segments to short ones. Note that there is no more than $|w|$ non-$\eps$-segments in any run and therefore there is no more than $|w|$ non-$\eps$-chains. We apply Lemma~\ref{ChainReplacementLemma} no more than $|w|$ times to get a run in which  
all query words are short and all $\eps$-segments are also short. Therefore this run is short.
\qed\end{proof}

\section*{Acknowledgments}
The authors are thankful to Dmitry Chistikov for the feedback and a discussion of the text's results,   and for suggestions for improvements. We are also much indebted to Mikhail Volkov and to anonymous referees for their helpful comments.



\bibliographystyle{splncs03}
\bibliography{sa_paper_springer}

\end{document}